%% file: main.tex
\documentclass[11pt]{article}
\usepackage[margin=1in]{geometry}

\usepackage{amsmath, amsthm, mathtools, dsfont, mathdots}
\usepackage{newtxtext}
\usepackage{newtxmath}
\usepackage{multirow}
\usepackage{booktabs} 
\renewcommand{\paragraph}[1]{\medskip\noindent\textbf{#1}}

\usepackage[hidelinks]{hyperref}

\usepackage{wrapfig}
\usepackage{makecell} 

\usepackage{xspace}
\usepackage{graphicx}

\usepackage{booktabs}
\usepackage{multirow}

\usepackage[linesnumbered,ruled]{algorithm2e}
\usepackage[noend]{algpseudocode}

\usepackage{array}
\usepackage{color}

\newcommand{\OPT}{\texttt{OPT}\xspace}

\newtheorem{assumption}{Assumption}
\newtheorem{definition}{Definition}

\newtheorem{theorem}{Theorem}

\newtheorem{lemma}{Lemma}
\newtheorem{proposition}{Proposition}

\DeclareMathOperator*{\argmax}{arg\,max}
\usepackage{mathtools}
\usepackage{comment}

\newcommand{\opt}{\texttt{OPT}\xspace}
\newcommand{\alg}{\texttt{ALG}\xspace}
\newcommand{\algfrac}{\texttt{ALG-FRAC}\xspace}

\newcommand{\alphathm}{1 + \ln\left(\frac{U}{L}\right)}

\newcommand{\cvar}{\ensuremath{\mathsf{CVaR}_{\delta}}\xspace}

\newcommand{\crcvar}{\ensuremath{\textsf{CVaR}_{\delta}\textsf{-CR}}\xspace}

\newcommand{\alphaFSP}{\alpha_{\delta}^{\textsf{SP}}\xspace}
\newcommand{\alphaFDP}{\alpha_{\delta}^{\textsf{DP}}\xspace}
\newcommand{\alphaDDP}{\alpha_{\delta}^{\Delta\textsf{-DP}}\xspace}

\newcommand{\oksrisk}{\texorpdfstring{\textsc{ROS-$(\delta,\Delta)$}}{kSelection-(delta,Delta)}\xspace}

\newcommand{\algname}{\textsf{cPPM}-$\boldsymbol{\phi}$\xspace}

\usepackage[commandnameprefix=always]{changes}

\usepackage{soul}
\sethlcolor{yellow}

\usepackage{enumitem}
\usepackage{subcaption} 

\author{
    Hossein Nekouyan\thanks{University of Alberta. Email: \texttt{nekouyan@ualberta.ca}}\\
    \and
    Bo Sun\thanks{University of Ottawa. Email:
    \texttt{bo.sun@uottawa.ca}}\\
    \and 
    Raouf Boutaba\thanks{University of Waterloo. Email:
    \texttt{rboutaba@uwaterloo.ca}}\\
    \and
    Xiaoqi Tan\thanks{University of Alberta. 
    Email: \texttt{xiaoqi.tan@ualberta.ca}}
}

\begin{document}

\title{Risk-Sensitive Online Selection with\\ Bounded Adaptivity}

\maketitle

\begin{abstract}
Designing randomized online algorithms that perform reliably not only in expectation but also under unfavorable realizations of randomness is a fundamental challenge in online decision-making. In this paper, we study this challenge in online adversarial selection, where a decision maker allocates $k$ units of a resource to sequentially arriving buyers through posted prices. We focus on two intertwined considerations that are often overlooked simultaneously: tail-risk sensitivity and bounded adaptivity, where tail risk is measured using conditional value-at-risk (CVaR) and bounded adaptivity limits the number of allowable policy updates over time. Our main contribution is a correlated posted-price mechanism that uses a single random seed to coordinate pricing decisions across time. This correlation induces a monotonic ordering of pricing profiles across sample paths, improving lower-tail performance while respecting the adaptivity constraint. More broadly, our results highlight correlation as a mechanism for controlling tail risk in randomized online algorithms. Using this framework, we derive competitive guarantees for several regimes of the problem under both static and dynamic pricing. Our analysis develops a risk-sensitive randomized online primal-dual framework tailored to CVaR objectives and reveals a systematic trade-off between allowable adaptivity, risk sensitivity, and competitive performance. Experiments on real airline pricing data further illustrate the empirical impact of correlated pricing on welfare concentration and tail behavior.

\end{abstract}

\section{Introduction}

Randomization is a powerful technique for improving the \emph{expected} performance of algorithms. However, the performance of a randomized algorithm can vary substantially across different realizations of its internal randomness, and algorithms with strong expected guarantees may still exhibit poor behavior on unfavorable sample paths. Developing randomized algorithms that perform reliably under such tail realizations is therefore a fundamental challenge in algorithmic decision-making, particularly in sequential, combinatorial, and learning-based settings where randomization plays a central role. Consequently, risk-sensitive objectives, which penalize poor tail outcomes rather than optimizing solely for expected performance, have recently attracted significant attention in machine learning and theoretical computer science (e.g.,~\cite{ni2024risk,wang2023near,submodular,simchilevi2023stochastic,bandits2}).

In this paper, we study the design of risk-sensitive online algorithms for online resource allocation, specifically online selection problems. Our work is motivated by two intertwined considerations that are often studied in isolation: tail-risk sensitivity and bounded adaptivity. First, while randomized online algorithms often achieve strong expected competitive guarantees for social welfare or revenue, their tail behavior remains significantly less understood. In online resource allocation, poor lower-tail performance of randomized algorithms is a recognized yet underexplored issue~\cite{Dinitz2024}. In this work, we adopt Conditional Value-at-Risk (CVaR) as our primary performance metric to ensure robustness under unfavorable sample paths. Second, we seek to constrain arbitrary adaptivity in posted-price mechanisms (PPMs), which are widely used in online resource allocation settings ranging from combinatorial auctions~\cite{golrezaei2014} and prophet inequalities~\cite{prophetlimited,prophetK} to online adversarial selection~\cite{sun2024static,jazi2025posted}. While dynamically changing prices can improve competitive performance, unconstrained adaptivity may induce de facto price discrimination over time and lead to high operational costs~\cite{pricechangeCost,priceDiscri2017}. We therefore ask whether competitive randomized PPMs can remain effective while satisfying a price-change cap~$\Delta$, which limits the total number of allowable policy updates.

\begin{figure}
    \centering
    \includegraphics[width = 0.95\linewidth]{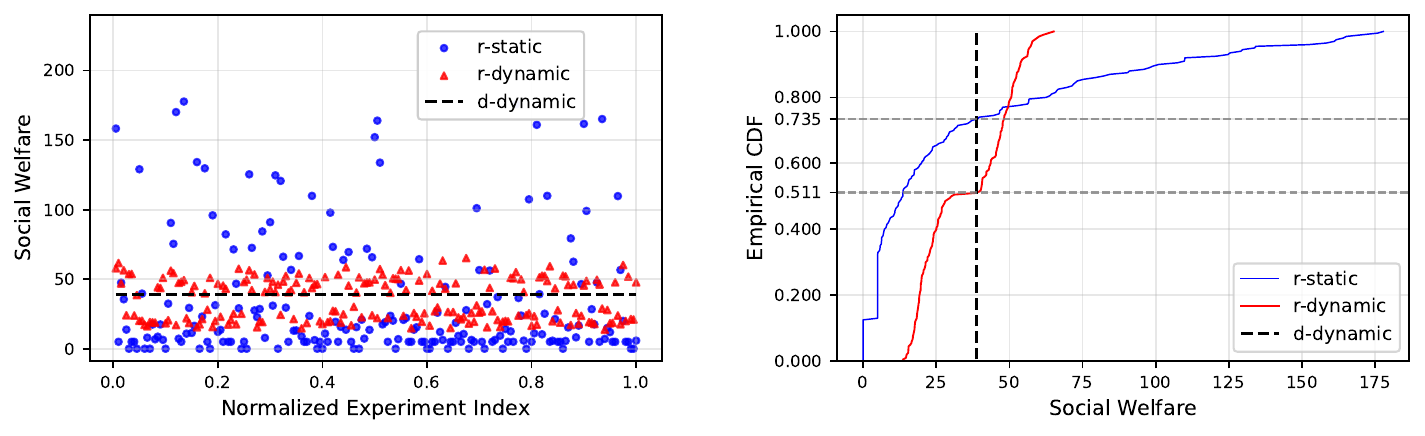}    
    \caption{Performance comparison of the \textsf{r-static}, \textsf{r-dynamic}, and \textsf{d-dynamic} algorithms on an instance of the \oksrisk problem. The \textsf{r-static} algorithm from~\cite{sun2024static} employs a single randomized price, the \textsf{r-dynamic} algorithm from~\cite{jazi2025posted} uses $k$ independent random seeds to generate $k$ randomized dynamically increasing prices, and the \textsf{d-dynamic} algorithm from~\cite{tan2023threshold} uses $k$ deterministic dynamically changing prices. The left plot shows performance over $10^4$ independent runs, and the right plot shows the empirical CDF.}
    \label{fig:intro} 
\end{figure}

In this work, we introduce \oksrisk, a variant of the online adversarial $k$-selection problem in which the objective is to optimize risk-sensitive social welfare under a price-change cap~$\Delta$. Specifically, we evaluate algorithms using \cvar of the total social welfare, where \cvar measures the expected performance over the worst $\delta$-fraction of sample paths. Existing works have made initial progress on risk-sensitive online algorithms and limited-adaptivity pricing, but largely treat these challenges separately. For example, recent work on limited adaptivity~\cite{sun2024static} focuses exclusively on static pricing ($\Delta = 0$). While such randomized static schemes achieve optimal expected guarantees and eliminate price discrimination, they often exhibit poor tail behavior. As illustrated in Figure~\ref{fig:intro}, the performance of \textsf{r-static} frequently falls below deterministic benchmarks and may even yield zero welfare with non-negligible probability. Conversely, dynamic randomized pricing schemes~\cite{jazi2025posted} can improve tail performance, but rely on essentially unconstrained adaptivity.

Our work takes a unified view of risk sensitivity and bounded adaptivity through the lens of correlation. In particular, we show how carefully correlating randomized pricing decisions across time can substantially improve lower-tail performance while respecting adaptivity constraints. This leads to a framework that reveals a fundamental trade-off between tail risk, allowable adaptivity, and competitive performance.

\begin{table}[t]
\centering
\footnotesize 
\setlength{\tabcolsep}{4pt}
\resizebox{\linewidth}{!}{%
\begin{tabular}{@{}l c c c c@{}}
\toprule
 & $ \delta $ &
$ \Delta = 0 $ & $ \Delta = k - 1 $ &
$ 1 \le \Delta \le k - 2$ \\
\midrule
\multirow{2}{*}{\textsc{\oksrisk}}
  & $\delta = 1$        
    & \checkmark~Optimal~\cite{sun2024static} 
    & \checkmark~Optimal~(\textbf{Theorem~\ref{thm:limited:price:change:optimality}}) 
    & \checkmark~Optimal~(\textbf{Theorem~\ref{thm:limited:price:change:optimality}}) \\
  & $\delta \in (0,1)$  
    & \checkmark~Optimal~(\textbf{Theorem~\ref{thm:design:cvar:static}})         
    & \checkmark~Optimal as $k \!\to\infty$~(\textbf{Theorem~\ref{thm:k:cvar:design:phi}}) 
    & Best-known~(\textbf{Theorem~\ref{thm:general:cvar:design:phi}}) \\
\addlinespace
\textsc{OSCC}  
  & $\delta = 1$      
    & \checkmark~Optimal~\cite{sun2024static} 
    & \checkmark~Optimal as $k \!\to\infty$~\cite{jazi2025posted}   
    & No known results \\
\addlinespace
\textsc{kSearch}
  & $\delta \in (0,1]$ 
    & Best-known~\cite{ChristiansonRisk2024} 
    & Optimal for $\delta = 1$,  as $k \!\to\infty$~\cite{kmaxSearch}  
    & No known results \\
\addlinespace
\textsc{IID-Prophet} 
  & $\delta = 1$ 
    & \checkmark~Optimal~\cite{prophetStatic} 
    & \checkmark~Optimal~\cite{prophetK} 
    & Best-known~\cite{prophetlimited} \\
\bottomrule
\end{tabular}
\par}%
\vspace{+0.2cm}
\caption{Summary of results for various online selection problems under different $(\delta, \Delta)$ settings and arrival models.
All results concern the design of posted-price or threshold-based algorithms.  The \textsc{OSCC} problem, studied in~\cite{sun2024static,jazi2025posted}, introduces a production cost associated with producing each additional unit of the item.
The \textsc{kSearch} problem, studied in~\cite{ChristiansonRisk2024,kmaxSearch}, is a variant of online selection where the decision maker is notified upon the arrival of the last buyer, thereby limiting the uncertainty about total demand.
In the \textsc{IID-Prophet} setting studied in \cite{prophetK,prophetStatic,prophetlimited}, buyers’ values are drawn independently from a common distribution known to the decision maker, allowing the use of distributional information in decision-making.}
\label{tab:optimality-summary}
\end{table}

\subsection{Our Contribution and Techniques}
This paper introduces the \oksrisk{} problem and develops a unified framework for studying risk-sensitive online selection under bounded adaptivity. Our theoretical contributions begin with the risk-neutral setting (i.e., $\delta = 1$), where we develop a family of posted-price mechanisms (PPMs) that achieve the optimal competitive ratio for every price-change cap $\Delta \in \{0,1,\dots,k-1\}$, significantly extending prior work on static pricing. We then turn to the risk-sensitive regime (i.e., $\delta \in (0,1)$). For the fully-static case (i.e., $\Delta = 0$), we characterize an optimal risk-sensitive static pricing scheme and establish the best possible $\cvar$-competitive ratio among single-price algorithms. Finally, we develop a general $\Delta$-level pricing framework for arbitrary adaptivity and risk levels. This framework reveals a systematic trade-off between allowable adaptivity, tail-risk sensitivity, and competitive performance, yielding monotone improvements in competitiveness as either $\Delta$ or $\delta$ increases. In the fully-dynamic regime ($\Delta = k-1$), our framework further attains asymptotically optimal performance in the large-inventory limit ($k \to \infty$). Table~\ref{tab:optimality-summary} summarizes our main results and compares them with existing works; additional related work is deferred to Appendix~\ref{apx:related:work}.

Our technical contributions center on two interrelated ideas. First, we develop a correlated PPM in which pricing decisions across time are coupled through a single random seed. This correlation induces a monotonic ordering of pricing profiles across sample paths, substantially improving lower-tail performance while simultaneously serving as a lossless online rounding scheme for fractional allocations. More broadly, our results highlight correlation as a structural mechanism for controlling tail risk in randomized online algorithms. Second, we establish performance guarantees through a novel risk-sensitive randomized online primal-dual (R-OPD) framework. This framework utilizes a dual program tailored to the $\Delta$-capped setting and restricts dual updates to the worst $\delta$-fraction of sample paths, thereby aligning the dual objective with the algorithm's $\cvar$ performance. The resulting analysis naturally leads to systems of delay differential equations~\cite{DDE}, capturing a memory effect induced by tail realizations and correlated sample paths.

Beyond the specific setting studied in this paper, both the correlated pricing scheme and the risk-sensitive R-OPD framework appear broadly applicable. We expect these techniques to provide useful tools for studying risk sensitivity and adaptivity constraints in a wider range of  online decision-making problems.

\section{Problem Setting}
\label{sec:setting}
We introduce the problem of \oksrisk as follows: A seller has $k$ identical units of an item, and faces $T$ buyers arriving one by one. When buyer $t$ arrives, the seller posts a price $p_t$. Buyer $t$ has a private value $v_t$ and accepts the price if $v_t \ge p_t$; otherwise the buyer leaves. The price may change at most $\Delta$ times over the horizon, i.e., $\sum_{t=1}^{T-1} \mathbf{1}_{\{p_t \ne p_{t+1}\}} \le \Delta$, where $ \Delta \in \{0, 1, \cdots, k-1\}$.

Let $x_t \in \{0,1\}$ indicate whether buyer~$t$ purchases an item. 
Buyer~$t$’s utility is given by $u_t = (v_t - p_t) x_t$, the seller’s revenue is $r = \sum_{t=1}^T p_t x_t$, and the total social welfare, which is defined as the sum of buyer utilities and seller revenue, is 
$r + \sum_{t=1}^T u_t = \sum_{t=1}^T v_t x_t$.
Let $P = \{p_t\}_{t=1}^{T}$ denote the vector of prices posted by an online algorithm~$\alg$.
In the online setting, the seller must determine each price~$p_t$ without knowing the values of future arrivals $\{v_{t'}\}_{t' > t}$ or even the total number of arrivals~$T$.
Since the algorithm may randomize its pricing decisions to manage uncertainty in buyer values and total demand, $P$ is treated as a random vector.

For an instance $I = \{v_t\}_{t=1}^T$ of \oksrisk, let $\alg(I, P)$ denote the random variable representing the total social welfare achieved by algorithm~$\alg$ on instance~$I$ under the random price vector~$P$. Let $F_{\alg(I, P)}$ be the cumulative distribution function (CDF) of this random variable.
We use CVaR as our risk metric as it is tail-sensitive, coherent, and convex, which makes it tractable for optimization~\cite{bookFinance,stochasticBook}.
Following the standard definition in~\cite{bookFinance,stochasticBook}, we define \cvar\ as follows:
\begin{definition}
Given $ \delta \in (0,1]$ and for a reward-type random variable~$X$, we define \cvar as
\begin{align*}
    \cvar[X] =\sup_{\tau \in \mathbb{R}} \Big\{ \tau - \tfrac{1}{\delta}\,\mathbb{E}\big[(\tau - X)_{+}\big] \Big\},
\end{align*} 
where $\delta $ specifies the risk level (tail probability) and $(x)_{+} = \max\{x,0\}$.
\end{definition}

We adopt the reward-based formulation of \cvar since the objective in this work represents a reward (social welfare) rather than a loss.
Moreover, if the CDF $F_{\alg(I, P)}$ corresponding to the random objective value of an online algorithm~$\alg$ is strictly increasing and continuous, the \cvar performance of that algorithm is given by
$\cvar[\alg(I,P)]
=\frac{1}{\delta} \int_{0}^{\delta} F_{\alg(I,P)}^{-1}(\eta)\, d\eta,
$ where $F^{-1}$ denotes the inverse cumulative distribution (quantile) function.
Intuitively, \cvar measures the algorithm’s expected performance over the worst $\delta$-fraction of its sample paths.

The objective is to design an online algorithm that minimizes its \cvar-competitive ratio, denoted by $\crcvar$, defined as
\begin{align*}
\crcvar(\alg) 
= \sup_{I \in \mathcal{I}} 
\frac{\opt(I)}{\cvar[\alg(I,P)]},
\end{align*}
where $\opt(I)$ denotes the offline clairvoyant optimum computed as the summation of the $ k$ highest {values}, 
i.e., $ 
\opt(I) = \max_{x_t} \sum_{t=1}^T v_t x_t  $, subject to $ \sum_{t=1}^T x_t \le k, \; x_t \in \{0,1\}, \; \forall\, t $.

Without any additional information regarding buyers’ values, no online algorithm can attain a bounded competitive ratio for \oksrisk, {even in the risk-neutral case with full-adaptivity \cite{sun2024static}}. We therefore follow the standard assumption in the literature that buyers’ values lie within a known bounded range.

\begin{assumption}\label{ass:bounded-value}
In \oksrisk, buyer values satisfy $v_t \in [L,U]$ for all $t \in [T]$, and we denote the ratio $U/L$ by $\theta$.
\end{assumption}
In the following section, we present a PPM that specifies how the posted prices at different price levels are generated and correlated across these levels.

\subsection{The Algorithm: Correlated PPMs with Limited Price Changes}
\label{sec:alg}

\begin{figure}     
\centering\includegraphics[width=0.8\linewidth]{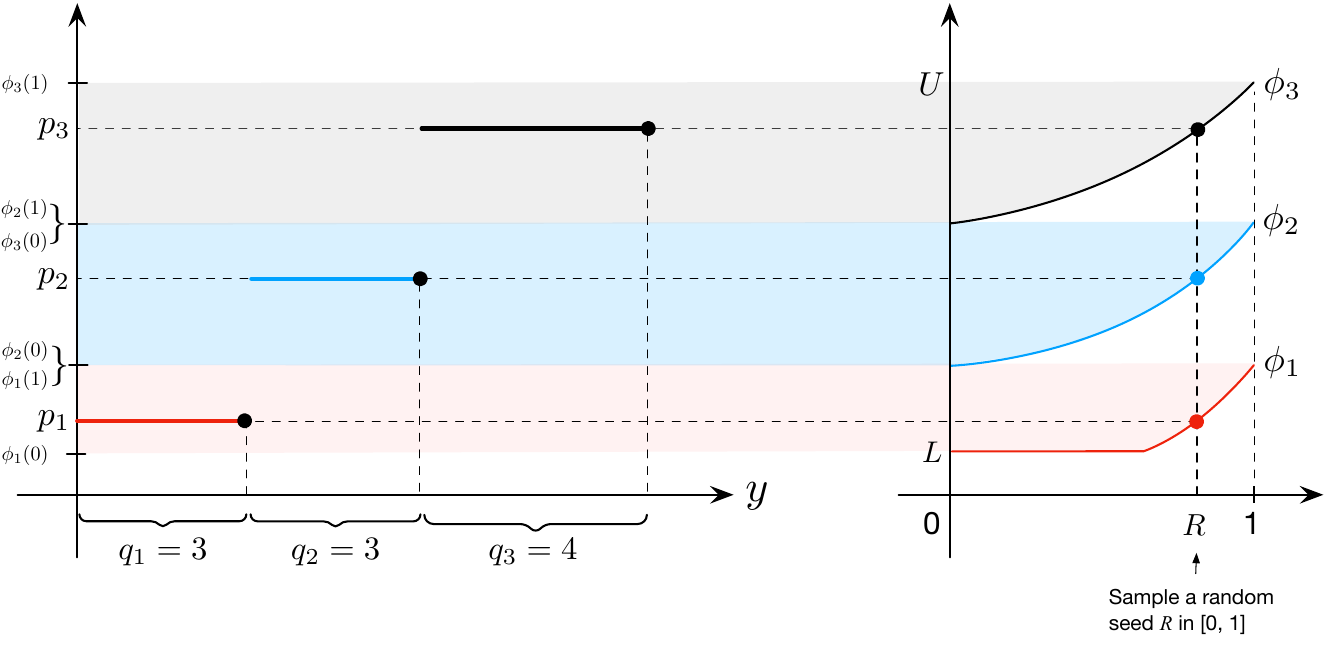}
  \caption{Illustration of \algname with $ \Delta = 2 $ (i.e., dynamic pricing with three total price levels), total units $ k = 10 $, and reservation vector $ \{q_1 = q_2 = 3,\, q_3 = 4\}$. When a random seed $ R \sim \mathcal{U}(0,1) $ is sampled, the three prices $ p_1 $, $ p_2 $, and $ p_3 $ are generated according to the pricing functions $ \phi_1 $, $ \phi_2 $, and $ \phi_3 $, respectively. By construction, the pricing functions satisfy $ L = \phi_1(0) \leq \phi_1(1) = \phi_2(0) \leq \phi_2(1) = \phi_3(0) \leq \phi_3(1) = U $, which ensures that $ p_1 \leq p_2 \leq p_3 $ always holds.}
    \label{fig:cPPM_illustration} 
\end{figure}

We introduce a correlated PPM, denoted by~\algname, which is formalized in Algorithm~\ref{alg:corr:PPM:limited:price:change}. 
{The mechanism is defined by a reservation vector $\{q_i\}_{i=1}^{\Delta + 1}$ and a corresponding set of $\Delta + 1$ pricing functions $\boldsymbol{\phi} := \{\phi_i\}_{i=1}^{\Delta + 1}$. 
Each $q_i$ represents the number of units reserved to be sold at price level~$i$ that is determined by the pricing function~$\phi_i$. 
Each pricing function $\phi_i$ is nondecreasing on the interval~$[0,1]$, and higher price levels dominate the lower ones throughout the entire range, i.e., $ \phi_i(0) \leq \phi_i(1) \leq \phi_{i+1}(0) \leq \phi_{i+1}(1),\; \forall\, i \in [\Delta]$.}

The mechanism begins by sampling a single random seed~$R$ uniformly from the interval~$[0,1]$. Using this seed, it correlates the prices across different price levels by setting the posted price at level~$j$ to $\phi_j(R)$ for each $j \in [\Delta + 1]$. See Figure~\ref{fig:cPPM_illustration} for an illustration of \algname when $\Delta = 2$.

The algorithm starts by allocating units reserved for the first price level $\phi_1(R)$. For each arriving buyer~$t$, it determines the corresponding price level $j_t$ based on the number of sold units $y_t$, and posts the price $p_t = \phi_{j_t}(R)$.
The buyer $t$ accepts the price if $v_t \ge p_t$, and declines the price otherwise. Then the mechanism updates the number of sold units accordingly.

In contrast to the algorithm proposed in \cite{jazi2025posted}, which samples an independent random price for each unit of the item,
Algorithm~\ref{alg:corr:PPM:limited:price:change} employs a single random seed to correlate prices across all levels. This ensures that pricing profiles $\{\phi_j(R)\}_{j \in [\Delta+1]}$ become gradually more aggressive as $R$ increases from $0$ to $1$, imposing a global order on the pricing profiles across sample paths. This in turn induces a monotonic behavior in resource utilization that simplifies the analysis. More importantly, this correlation synchronizes posted prices such that the total welfare varies only slightly across sample paths. As a direct consequence, when evaluating the \cvar objective, the algorithm incurs a minimal loss, since the remaining $(1-\delta)$-fraction of outcomes does not yield significantly better total welfare than the worst $\delta$-fraction. 

The core of \algname lies in the design of pricing function $\phi$. In the following sections, we study the design and analysis for different variants of the \oksrisk\ problem. We refer to \algname with a single pricing function (i.e., $\Delta = 0$) as \textit{fully-static pricing}, with $k$ pricing functions (i.e., $\Delta = k-1$) as \textit{fully-dynamic pricing}, and for any $\Delta = 1, \cdots, k-2$ as the \textit{$\Delta$-dynamic pricing}.

\begin{algorithm}[H]
\caption{Correlated PPM with Pricing Functions $ \boldsymbol{\phi}$ (\algname)}
\label{alg:corr:PPM:limited:price:change}
\KwIn{A set of pricing functions $ \boldsymbol{\phi} = \{\phi_i\}_{i=1}^{\Delta + 1}$, reservation vector $\{q_i\}_{i=1}^{\Delta + 1}$}
\textbf{Initialize:} $y_1 \gets 0$

Sample a random seed $R \sim \mathcal{U}(0,1)$

\For{each buyer $t = 1, 2, \dots$}{
\If{$y_t < k$}{
    Let $j_t \gets \max \left\{ j \in [\Delta + 1] \,\middle|\, y_t \ge \sum_{l=1}^{j - 1} q_l \right\}$\,

    Post price $p_t = \phi_{j_t}(R)$ to buyer $t$\,

    \eIf{buyer $t$ accepts the price $p_t$ (when $v_t \ge p_t$)}{
        Set $x_t \gets 1$\,
    }{
        Set $x_t \gets 0$\,
    }

    Update $y_{t+1} \gets y_t + x_t$\,
}
}
\end{algorithm}

\section{Main Results}
\label{sec:main}
In this section, we present the main theoretical results concerning the
\cvar-competitive performance of \algname\ for the \oksrisk\ problem. We
progressively develop the pricing functions utilized by \algname, beginning
with the risk-neutral case ($\delta = 1$) as a warm-up, whose intuition and
design serve as a foundation for the analysis that follows, and subsequently
moving toward the general risk-aware formulation. Each subsection introduces
the motivation for the specific design, states the key result, and discusses
the underlying analytical structure when applicable.

\paragraph{Risk-Neutral Posted Pricing with Limited Price Changes.}
We begin with the risk-neutral case ($\delta = 1$), where the \cvar objective
reduces to the standard expected-performance guarantee. While \cite{sun2024static}
establishes an optimal static pricing algorithm for $\Delta = 0$, their analysis
leaves open the question of how to design posted-price schemes that achieve
optimal competitive ratios when a higher number of price changes is permitted.
Studying this question in the simpler risk-neutral setting provides structural
insights that guide the design and analysis of posted prices in the more
challenging risk-sensitive and bounded-adaptivity settings. The following
theorem provides the optimal design for \algname\ that achieves the tightest
possible \cvar-competitive ratio for the \oksrisk\ problem.
\begin{theorem}
\label{thm:limited:price:change:optimality}
Consider \oksrisk with $ \delta = 1 $ and any given price-change cap $\Delta \in \{0,1,\dots,k-1\}$.
{Let 
$\{q_j\}_{j \in [\Delta+1]}$ be any reservation vector satisfying 
$q_1 \le q_2 \le \dots \le q_{\Delta+1}$ and $\sum_{j=1}^{\Delta+1} q_j = k$}. 
\algname is $ \alpha^{\star}$-competitive, where {$ \alpha^{\star} = 1 + \ln(\theta) $}, if for all $j \in [\Delta+1]$, $ \phi_j $ is given by
\begin{align} \label{eq:phi:design:limited:price}
\phi_j(x) = 
\begin{cases}
L  &\quad \text{if } 
\dfrac{\sum_{l=1}^{j-1} q_l + q_j x}{k} \in \left[0, \tfrac{1}{\alpha^{\star}}\right), \\
L \cdot \exp\!\Big( \alpha^{\star} \cdot \dfrac{\sum_{l=1}^{j-1} q_l + q_j x}{k} - 1 \Big)  
&\quad \text{if } 
\dfrac{\sum_{l=1}^{j-1} q_l + q_j x}{k} \in \left[\tfrac{1}{\alpha^{\star}}, 1\right].
\end{cases}
\end{align}
\end{theorem}

In the above design of pricing functions, we require the reservation vector to be non-decreasing from index $1$ to $\Delta+1$, meaning that more units are reserved at higher price levels.
This condition ensures that across different sample paths, the number of price levels whose reserved units are fully exhausted differs by at most one (see Lemma~\ref{lem:lower-bound:yr}). This is a key structural property underlying our primal-dual analysis. Since the number of exhausted price levels governs how many units are sold and at what prices, having this differ by at most one limits the oscillations in the total welfare across sample paths. This in turn reduces the risk-sensitivity of \algname. From a practical
perspective, this condition is also natural: to maximize social welfare in practice, one typically prefers to sell more units at higher price levels, since buyers at those levels generate greater surplus.

\begin{proof}[Proof Sketch of Theorem \ref{thm:limited:price:change:optimality}]
We utilize the randomized online primal-dual (R-OPD) framework for this proof.
We formulate a dual linear program
that upper-bounds the offline clairvoyant
optimum, with one set of dual variables tracking buyer utilities and another
tracking prices at each price level. For each realization of the random seed,
dual variables are updated
so that the accumulated dual objective exactly equals
the algorithm's realized social welfare. To do so, we need to identify, for
each realization of the random seed, how the reserved units at each price level
are allocated and which price levels are fully exhausted.

To this end, we establish that {(i) the number of units sold is non-increasing in the random seed $R$ (see Lemma~\ref{lem:monotonicity:1}), since the pricing profiles become more aggressive as $R$ increases, and (ii) for any realization of the random seed, the algorithm is guaranteed to sell all reserved units up to a certain price level (see Lemma~\ref{lem:lower-bound:yr}), which follows from the correlated rounding scheme and the monotonicity assumption on the reservation vector.}
These two properties together allow us to determine exactly which price levels are exhausted under each realization, and thereby construct the dual updates consistently.
The correlated pricing scheme is what makes this possible: because all price levels are generated from a single random seed, sample paths are globally ordered by aggressiveness, making the utilization structure tractable.

Using the pricing functions given in Eq.~\eqref{eq:phi:design:limited:price}, one can verify that the resulting dual solution is feasible up to the factor~$\alpha^\star$.
By weak duality, \textup{cPPM-}$\phi$ achieves at least a $1/\alpha^\star$ fraction of the offline optimum.
Since $\alpha^\star = 1 + \ln\theta$ matches the established lower bound~\cite{sun2024static}, the competitive ratio is optimal for every $\Delta \in \{0,1,\ldots,k-1\}$.
See Appendix~\ref{sec:expected} for more intuitions and the full proof.
\end{proof}

\paragraph{Risk-Sensitive Fully-Static Pricing.}  To build-up the results for the general risk-sensitive case of \oksrisk, we first consider the fully-static case with no price changes, i.e., $\Delta = 0$.  Below, we derive the static pricing that obtains the tightest \crcvar\ across all static pricing algorithms.

\begin{theorem}[\textsc{Risk-Sensitive Fully-Static Pricing}]
\label{thm:design:cvar:static}
Consider \oksrisk with  $ \delta \in (0, 1] $ and $\Delta=0$.  \algname achieves the optimal~\crcvar, denoted by  $\alphaFSP $, among all fully-static pricing schemes if $q_1 =  k$ and the single price function $\phi$ is given by
\begin{align*}
\phi_1(x) =
\begin{cases}
L & x \in \big[0, \tau_\delta\big] \\ 
\displaystyle
L\Big[1 + \sum\nolimits_{j=1}^{N(x)}
\frac{\left({\alphaFSP}/{\delta}\right)^{j}}{j!}
\left( x - 1 + \delta(1 - \tfrac{1}{\alphaFSP}) - (j-1)\cdot (1 - \delta) \right)^{j}
\Big] & x \in \big[\tau_\delta, 1\big],
\end{cases}
\end{align*}
where $\alphaFSP$ is the unique solution to the equation $\phi_1(1) = U$, $ N(x)
=
\left\lfloor
\frac{x-1+\delta-\delta/\alphaFSP}{1-\delta}
\right\rfloor+1$ and $\tau_\delta =  1 - \delta + \delta \cdot  \frac{1}{\alphaFSP}$. 
\end{theorem}

Since $\phi_1(1)$ is monotonically increasing with respect to~$\alphaFSP$, 
the equation $\phi_1(1) = U$ admits a unique solution. 
As $\delta \rightarrow 1$, the value of $\alphaFSP$ converges to $1 + \ln \theta $.
Consequently, the fully-static pricing scheme designed according to Theorem~\ref{prop:cvar:static:design:phi} coincides with the optimal pricing design of the online algorithm for the \oksrisk  problem when $\delta = 1$ and $\Delta = 0$, as established in Theorem~\ref{thm:limited:price:change:optimality}. 

\begin{proof}[Proof Sketch of Theorem \ref{thm:design:cvar:static}.]
The design of $\phi_1$ in Theorem~\ref{thm:design:cvar:static} is motivated
by identifying the hard family of instances for fully-static pricing
algorithms (see Proposition~\ref{prop:lb:static}).
These instances consist of buyers arriving in groups of $k$,
where all buyers within a group share the same value, and values
increase across groups.
Since the sequence may stop after any group, the
algorithm cannot distinguish between instances that differ only in when the
sequence terminates, and must therefore post a price that performs well
across all possible stopping points simultaneously.

We show that any static price distribution must satisfy a set of constraints in order to be $\alpha$-\crcvar competitive on this class (see Proposition~\ref{prop:lb:static}). These constraints yield a system of delay differential equations that characterizes the optimal pricing function on these hard instances. Solving this system for
the smallest possible competitive ratio yields a lower bound on the smallest \crcvar achievable by any fully-static algorithm for \oksrisk .

Following a bottom-up approach~\cite{sun2024static}, \algname uses the optimal pricing function obtained from this lower-bound construction as its static pricing rule, whose closed-form solution is given in Theorem~\ref{thm:design:cvar:static}. The matching upper bound is then
proved using a risk-sensitive R-OPD argument (see Proposition~\ref{prop:cvar:static:design:phi}), where dual variables are
updated only over the worst $\delta$-fraction of seed realizations, so that the expected dual objective matches the \cvar objective of \algname. The pricing function design obtained from the lower-bound analysis ensures that the resulting expected dual solution is feasible up to the factor $\alpha^{SP}_{\delta}$ (see Proposition~\ref{prop:cvar:static:design:phi}). By weak duality, this gives $\alpha^{SP}_{\delta}$-competitiveness. See Appendix~\ref{apx:static} for the intuition behind this pricing design and the full proof.
\end{proof}

\paragraph{Risk-Sensitive $ \Delta$-Dynamic Pricing: A General Framework.} Building on the insights from the above two special cases, we extend the framework to the general risk-sensitive case with $\Delta \ge 1$.  The pricing functions are now recursively determined through a system of delay differential equations  that capture the effect of the tail probability~$\delta$ and correlation among pricing levels.

\begin{theorem}[\textsc{Risk-Sensitive $\Delta$-Dynamic Pricing}]
\label{thm:general:cvar:design:phi}
Consider \oksrisk with {$\delta \in (0,1)$} and any number of price changes $\Delta \in \{1,\dots,k-1\}$.
Let $\{q_j\}_{j \in [\Delta+1]}$ be a reservation vector satisfying $q_1 = \lceil \tfrac{k}{\alphaDDP} \rceil$ and $q_2 \le \dots \le q_{\Delta+1}$.
Then the \crcvar of \algname is $\alphaDDP$ if the following two conditions hold.
\textnormal{(i)} $\alphaDDP \ge 1$ is the unique solution of $\alphaDDP
=
\frac{k U}{ \frac{1}{2 \cdot \delta}
\sum_{j=1}^{\Delta+1} q_j \cdot \int_0^\delta\,\phi_j(\eta)\,d\eta }.$
\textnormal{(ii)} Let $\tau := 1-\delta$ and $c := \frac{\alphaDDP}{2k\delta}$.
The pricing functions $\boldsymbol{\phi}=\{\phi_i\}_{i\in[\Delta+1]}$ are recursively designed as follows: set $\phi_1(x)=L$ for all $x\in[0,1]$; and for each $i\in\{2,\dots,\Delta+1\}$, set $ \phi_i $ according to
\begin{equation}
\label{eq:phi_piecewise_fde_general}
\phi_i'(x)
=
\begin{cases}
\displaystyle
c\sum\nolimits_{j=2}^{i-1} q_j\bigl(\phi_j(x+\delta)-\phi_j(x)\bigr),
& x \in [0, \tau]
\\[3mm]
\displaystyle
c\big(
\sum\nolimits_{j=2}^{i-1} q_j\bigl(\phi_j(x-\tau)-\phi_j(x)\bigr)
+
q_i\,\phi_i(x-\tau)
\big),
& x \in (\tau, 1]
\end{cases}
\end{equation}
with initial value $
\phi_i(0)
=
c\big(
\lceil \tfrac{k}{\alphaDDP} \rceil L\delta
+
\sum_{j=2}^{i-1}\int_0^\delta q_j\,\phi_j(\eta)\,d\eta
\big).$
\end{theorem}

The proof of Theorem~\ref{thm:general:cvar:design:phi} is provided in Appendix~\ref{sec:general}. Here, we discuss the key intuition behind the design of the pricing functions $\boldsymbol{\phi}=\{\phi_i\}_{i\in[\Delta+1]}$ based on the system of \textit{delayed} ODEs in Eq. \eqref{eq:phi_piecewise_fde_general}.
Recall that the objective of the
algorithm is averaged only over the worst $\delta$-fraction of seed
realizations. Thus, the price chosen at a seed $x$ must be calibrated relative to the
prices at nearby worse seed realizations, which are typically shifted by
$1-\delta$. This creates a memory effect: the marginal change $\phi'(x)$ depends
on delayed values of the pricing function, such as $\phi(x-(1-\delta))$, rather
than only on its current value. Consequently, the system of equations that
determines the pricing design naturally takes the form of delay differential
equations. The fact that, at each seed $x$, we can trace back the worst $\delta$-fraction
of realizations and look at the prices posted at shifted locations is an effect
generated by the correlated scheme used by \algname. The correlated scheme
creates a monotonicity effect in the aggressiveness of the pricing profile, and
therefore creates monotonicity in the number of units sold across different seed
realizations. As a result, the worst $\delta$-fraction of realizations can be
identified as a contiguous subinterval of the seed range $[0,1]$, which is what
allows us to trace back the shifted locations.
More detailed intuitions  are provided in Appendix~\ref{sec:general}.

\textit{Case Study of Theorem \ref{thm:general:cvar:design:phi}: $\Delta=2$ with Balanced Reservation.}
We argue that solving Eq. \eqref{eq:phi_piecewise_fde_general} with closed-form expressions of $\boldsymbol{\phi} $ is generally impossible. Nevertheless,  for the special case of $\Delta = 2$, we derive such analytical results and obtain a tight upper bound for the \crcvar competitiveness of \algname up to a constant factor. Consider \oksrisk\ with $\delta \le \tfrac12$, and suppose the reservation vector is balanced across the last two price levels, namely, $q_2=q_3$.
In this case, there exists a pricing scheme according to Theorem~\ref{thm:general:cvar:design:phi} for \algname\ whose \crcvar\ is $\alphaDDP$, where $\alphaDDP$ admits the following approximation up to absolute constant factors: 
$$
\alphaDDP \lesssim
\begin{cases}
\sqrt{\theta}, & 0<\delta\le \frac13,\\
\frac{\delta}{3\delta-1}\theta^{1/3}+1, & \frac13<\delta\le \frac12 .
\end{cases}
$$ 
Details of the derivation are provided in Appendix~\ref{apx:corrolary:theorem:4}.

\begin{wrapfigure}{r}{0.4\textwidth}
    \centering
    \includegraphics[width=0.95\linewidth]{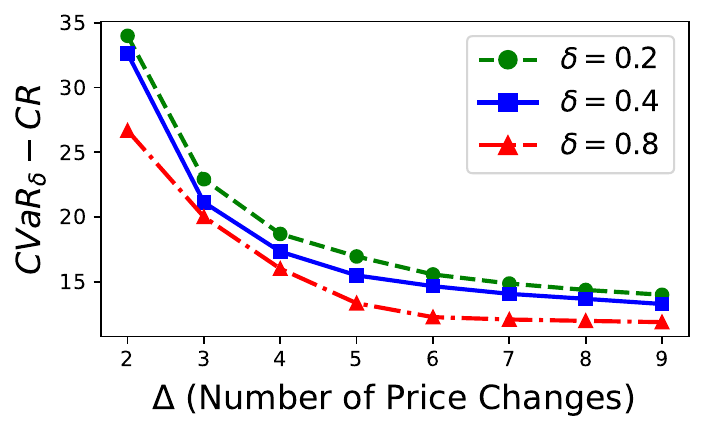}
    \caption{\crcvar of \algname, where $\phi$ is designed based on Thm~\ref{thm:general:cvar:design:phi}.}
    \label{fig:general}
\end{wrapfigure}

\textit{Numerical Results of Theorem \ref{thm:general:cvar:design:phi}.} From a numerical perspective,  established methods for solving delay differential equations have shown that  each pricing function $\phi_i$ can be computed recursively from the preceding pricing functions $\phi_{i'}$ for all indices $i' < i$ \cite{DDE}.
In particular, the smallest feasible value of $\alpha$ that admits a feasible design of the pricing functions satisfying Theorem \ref{thm:general:cvar:design:phi} can be computed via binary search. Figure~\ref{fig:general} illustrates the worst-case \crcvar of \algname, where the pricing functions are designed according to the theorem above, for three representative cases with $\delta \in \{0.2, 0.4, 0.8\}$ and the units are uniformly distributed among different price levels. 
The curves correspond to the setting with $L = 1$, $U = 100$, and $k = 40$. As observed, increasing the number of allowed price changes leads to tighter performance guarantees. Moreover, for higher allowable risk levels (i.e., larger values of~$\delta$), \algname achieves a better worst-case competitive ratio, since the algorithm becomes less risk-sensitive and can therefore adopt a more aggressive pricing design.

\paragraph{Risk-Sensitive Fully-Dynamic Pricing.}
For the extreme case $\Delta = k - 1$, where the price-change constraint in \oksrisk is fully relaxed, we can provide a stronger pricing design than the one in Theorem~\ref{thm:general:cvar:design:phi}, which in turn yields an optimality result:

\begin{theorem}
    [\textsc{Risk-Sensitive Fully-Dynamic Pricing}]
\label{thm:k:cvar:design:phi}
Consider \oksrisk with $\delta \in (0,1)$ and $\Delta = k-1$. Then there exists a pricing-function design
$\boldsymbol{\phi}=\{\phi_i\}_{i\in[k]}$ such that \algname\ achieves a \crcvar\ of $\alphaFDP$, where
$\alphaFDP \ge 1$ is the unique value satisfying $\alphaFDP = \frac{kU\delta}{\sum_{i=1}^{k} \int_{0}^{\delta} \phi_i(\eta)\, d\eta}$.
Moreover, the \crcvar of
\algname is asymptotically optimal as $k \rightarrow \infty$, namely, \algname  attains the smallest possible \crcvar among all online algorithms for any confidence level $\delta \in (0,1)$.
\end{theorem}

\begin{proof}[Proof Sketch of Theorem \ref{thm:k:cvar:design:phi}.]
In the fully-dynamic setting, the proof departs from the primal-dual framework
and instead interprets \algname as a lossless online rounding scheme.
We
compare \algname with a fractional algorithm that uses the same pricing
functions. The key property is that, for each buyer, the probability that
\algname allocates a unit to that buyer is equal to the fractional allocation
made by this fractional algorithm (see Lemma~\ref{lemma:ppm:frac:relation}). This equivalence is induced by the correlated
pricing scheme, which generates all prices using a single random seed.

Therefore, it is sufficient to analyze the fractional algorithm: once its
performance is upper-bounded against optimal clairvoyant algorithm, the same bound transfers to
\algname because \algname closely tracks the fractional allocation. In
addition, the single-seed correlation creates a monotonicity property: as the
seed increases, pricing profiles correlated through the random seed become more aggressive and the number of allocated units
decreases monotonically (see Lemma~\ref{lem:monotonicity:1}). This makes the \cvar analysis tractable,
because the worst $\delta$-fraction of sample paths can be identified as a
subinterval of the seed range $[0,1]$. Combining these observations yields the
claimed \crcvar guarantee for the fully-dynamic pricing
scheme. The detailed construction of the pricing functions is again recursive and follows a system of delay differential equations. The details regarding pricing design and the full proof of the theorem is deferred to Appendix~\ref{apx:sec:full:dynamic}.
\end{proof}

\begin{figure}
    \centering 
    \begin{minipage}[t]{0.48\linewidth}
        \centering
        \includegraphics[width=\linewidth]{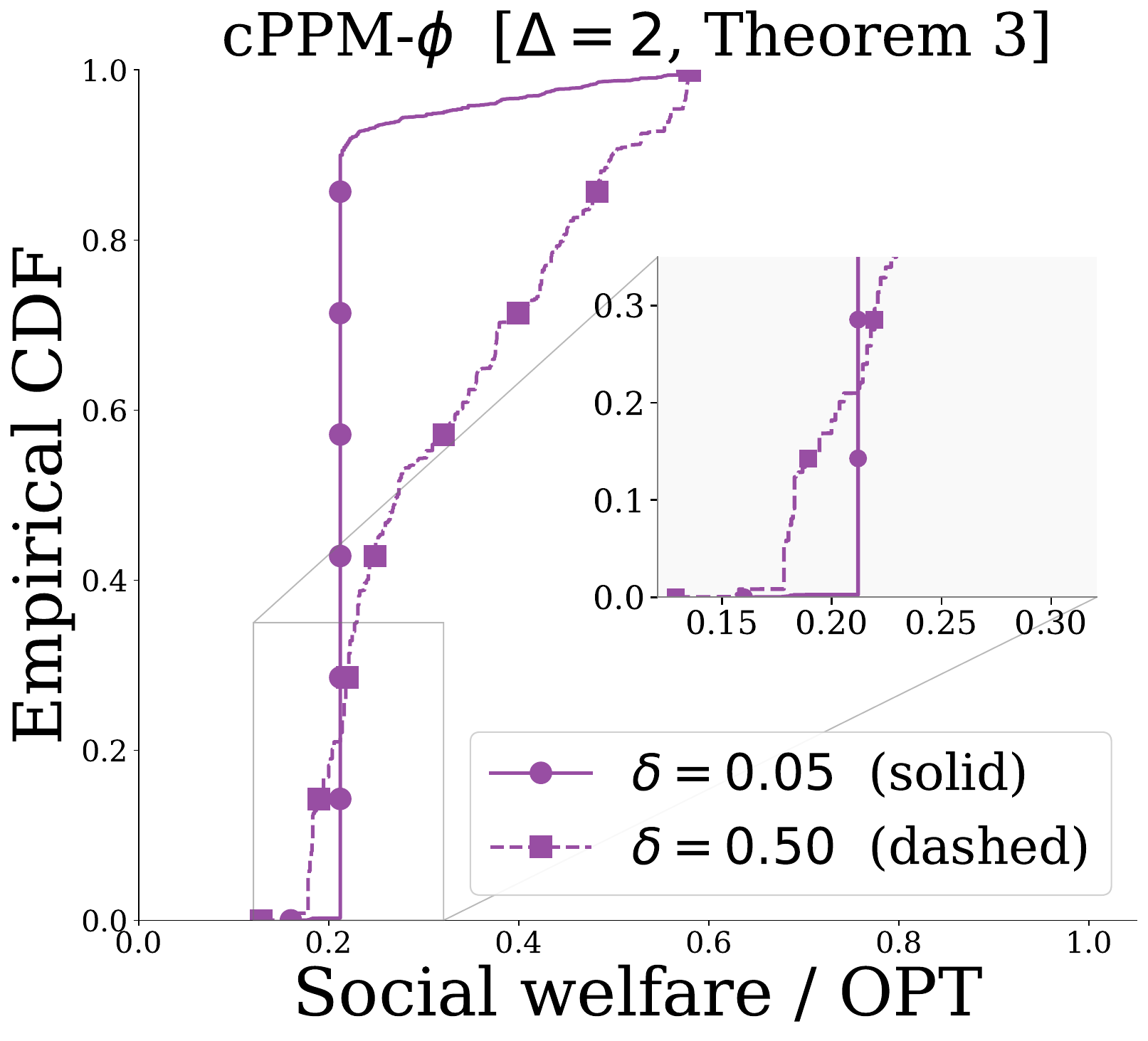}
    \end{minipage}%
    \hfill
    \begin{minipage}[t]{0.48\linewidth}
        \centering
        \includegraphics[width=\linewidth]{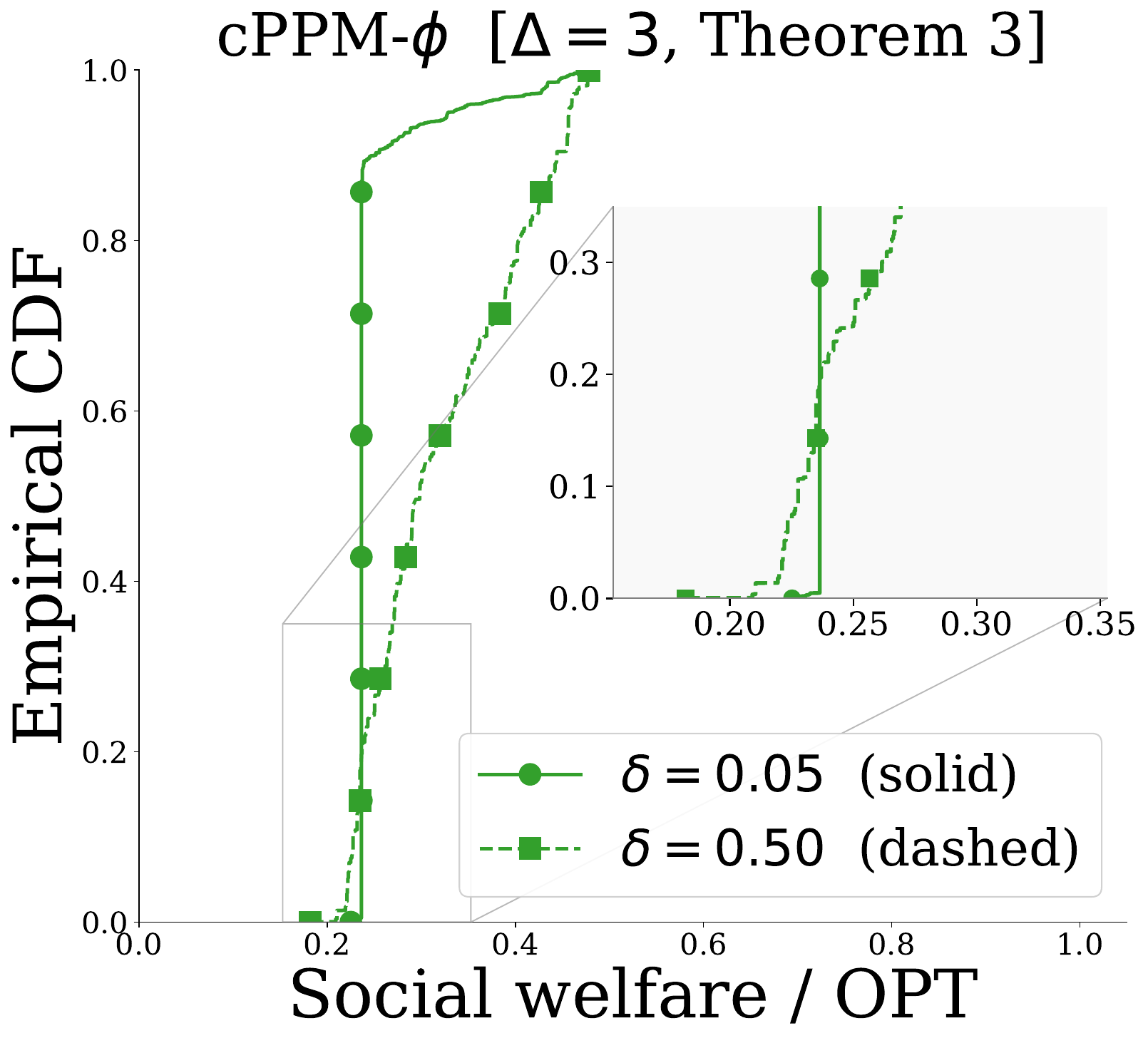}
    \end{minipage}

    \vspace{+1cm}

    \includegraphics[width=\linewidth]{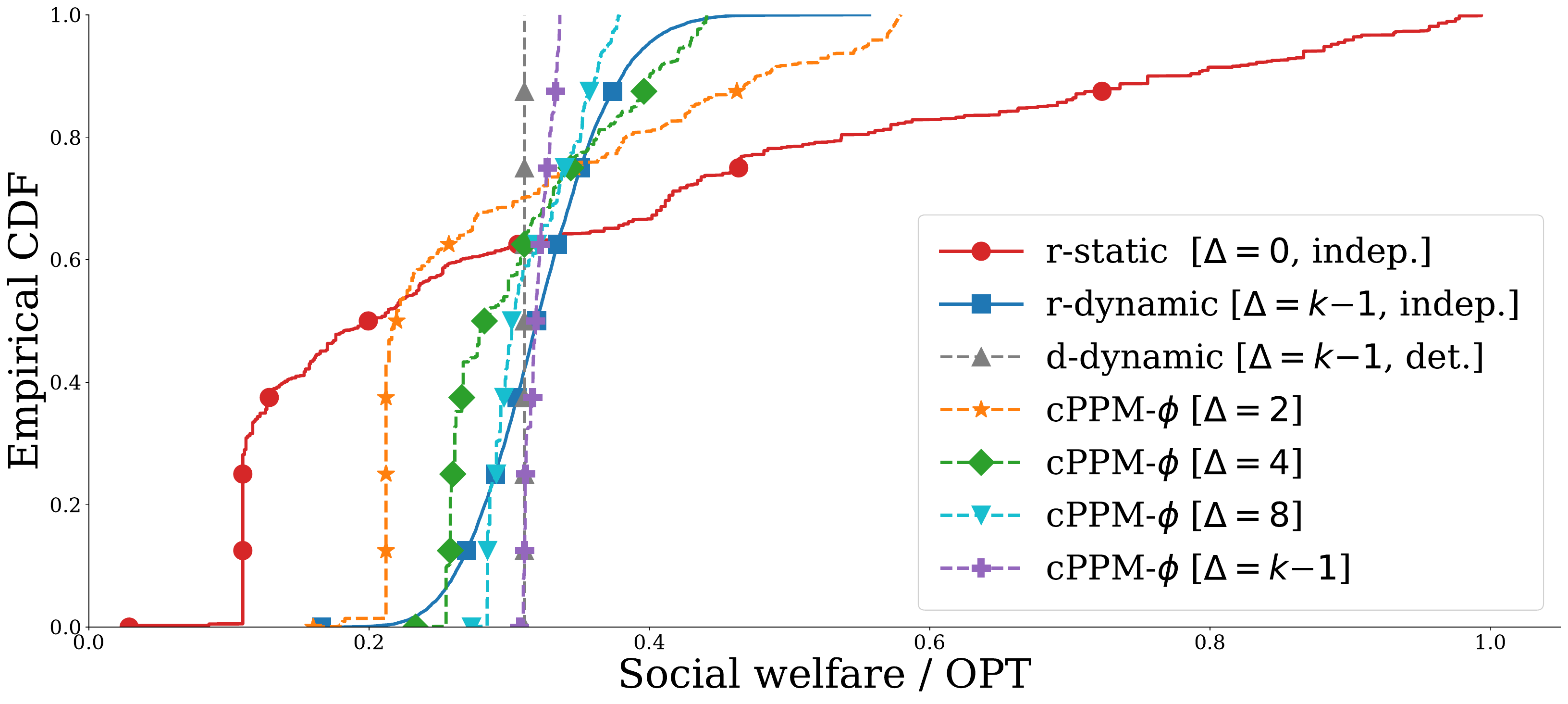}

    \caption{Empirical CDFs of performance ratios on the JFK--LAX Q1~2023 instance ($N=10{,}000$ runs).}
    \label{fig:experiment_combined}
\end{figure}






\section{Experiments}
\label{sec:experiment}

We evaluate \algname on real pricing data from the BTS Airline Origin
and Destination Survey (DB1B)~\cite{BTS_DB1B,Williams2022}, a publicly
available 10\% sample of domestic airline tickets.
To construct a single-product instance consistent with \oksrisk, we
restrict attention to one-way American Airlines JFK-LAX tickets in
Q1~2023, yielding 617 fares in $[L,U]=[\$140,\$1{,}996]$ with $U/L=14.3$.
We interpret each fare as the buyer value, and set $k=40$ and $T=380$.
Following standard revenue management
practice~\cite{TallurvanRyzin2004,Williams2022}, we sort values within
each instance in ascending order, reflecting the empirical pattern that
more price-sensitive leisure travelers tend to book earlier than less
elastic business travelers.
We draw one instance once and hold it fixed, varying only the internal
random seeds of the randomized algorithms across $N=10{,}000$ runs.
Throughout, the $x$-axis reports the ratio of the total social welfare of algorithms over optimal clairvoyant algorithm; this differs from the \crcvar studied
theoretically, which measures the ratio of optimal to achieved \cvar of social welfare.

The top two panels in Figure~\ref{fig:experiment_combined} compare the empirical CDFs of
\algname for $\Delta\in\{2,3\}$ under two risk levels
$\delta\in\{0.05,0.50\}$, where the pricing functions are according to
Theorem~\ref{thm:general:cvar:design:phi} and the units are equally distributed among different price levels.
When a stricter risk level is imposed (i.e., a smaller $\delta$), the pricing becomes more conservative. Consequently, the algorithm achieves stronger lower-tail performance but yields  a smaller expected social welfare. 
In addition, when adaptivity is more constrained (i.e., a smaller price-change cap $\Delta$), the pricing has less flexibility, and thus the algorithm attains weaker tail performance.

The bottom panel in Figure~\ref{fig:experiment_combined} shows the empirical CDFs of multiple benchmarks and \algname under a fixed risk level $\delta=0.30$. The results reinforce our previous observation that, with greater adaptivity (i.e., larger $\Delta$), the welfare achieved by \algname is more concentrated, with stronger tail performance and higher expected welfare.
Furthermore, the results show that a relatively small $\Delta$ already achieves a favorable balance: \algname with $\Delta=4$ closely tracks the fully-dynamic
variant ($\Delta=k-1$) in both tail and expected welfare. This suggests
that only a handful of price changes suffices to obtain near-optimal
risk-adjusted performance on this instance.
Finally, comparing \textsf{r-dynamic} \cite{jazi2025posted} (independent seeds, $k$ draws) against
\algname with $\Delta=k-1$ (single correlated seed), our algorithm
achieves strictly better lower-tail performance while attaining
comparable expected welfare. This highlights the importance of introducing correlation to improve the trade-off between risk and performance.

\section{Conclusion}
\label{sec:conclusion}
We investigated the design of posted-price mechanisms (PPMs) for online adversarial selection under two intertwined considerations: bounded adaptivity and risk-sensitive performance. Specifically, we studied pricing policies subject to a cap on the number of allowable price changes (denoted by $\Delta$) and evaluated performance using the \cvar objective, where $\delta$ captures tail risk. Our primary contribution is the development of \algname, a correlated posted-pricing framework that uses a single random seed to coordinate pricing decisions across time. This correlation induces a monotonic ordering of pricing profiles across sample paths, substantially improving lower-tail performance while respecting the adaptivity constraint. More broadly, our results highlight correlation as a structural mechanism for controlling tail risk in randomized online algorithms. Our theoretical analysis establishes competitive guarantees for several regimes of the problem and reveals a systematic trade-off between allowable adaptivity, risk sensitivity, and competitive performance.

Several open problems remain for future research. One direction is to develop a tighter analysis of the \crcvar competitiveness of \algname in the most general setting; in particular, we conjecture that the pricing design in Theorem~\ref{thm:general:cvar:design:phi} can be further refined to establish optimality. Another promising direction is to extend the correlation-based framework developed in this paper to other online decision-making problems, such as online matching and combinatorial auctions. It would also be interesting to investigate alternative correlation or rounding schemes, as well as other classes of risk-sensitive objectives beyond \cvar. Finally, an important practical direction is to study these problems in data-driven settings where value distributions are unknown and must be learned online.

\bibliographystyle{plain}
\bibliography{reference}

\newpage
\appendix

\input{appendix}

\end{document}

%% file: appendix.tex

\section{Further Related Work}

\paragraph{Single-Leg Revenue Management.}
The \oksrisk model can be interpreted as a single-leg revenue management problem, in which a seller allocates a fixed-capacity resource (e.g., airline seats or advertising impressions) over time. A substantial body of work (e.g., \cite{maSingleLeg,singleLegBall,singlelegAdvide}) studies such settings through competitive analysis, developing booking-limit and posted-price controls with provably optimal or near-optimal worst-case guarantees. Our framework is closely related to this line of work. In particular, we study \oksrisk through posted-price mechanisms that reserve resource units across multiple price levels and analyze their competitive guarantees relative to the offline optimum.

\paragraph{Pricing with Limited Price Changes.}
Prior work has studied dynamic pricing under constraints on the number of allowable price updates. In particular, \cite{priceDiscri2017} consider dynamic pricing with unknown demand and a limited number of price levels, designing policies that strategically balance exploration and exploitation while respecting a cap on price changes. Their motivation stems from the observation that, in many markets, firms cannot adjust prices arbitrarily frequently due to operational frictions, customer reactions, or adjustment costs~\cite{pricechangeCost}. Motivated by these considerations, several subsequent works have explored pricing problems with limited adaptivity in both theoretical and applied settings (e.g., \cite{chen2020LimitedPrice,orLimited}).

\paragraph{Risk-Aware Randomized Algorithms.}
Recent works have begun studying randomized online algorithms under tail-risk objectives. For example, \cite{Dinitz2024} study the online ski rental problem under a tail-risk constraint specified by the pair $(\gamma, \delta)$, designing algorithms such that the probability of the competitive ratio exceeding $\gamma$ is at most $\delta$. Another recent work, \cite{ChristiansonRisk2024}, studies several online problems (including ski rental and 1-max search) under the Conditional Value-at-Risk (CVaR) metric. More broadly, risk-sensitive algorithm design has attracted increasing attention in machine learning and learning theory, where CVaR-based objectives have been explored in settings including Markov decision processes~\cite{MDPrisk,MDPrisk2}, online submodular optimization~\cite{submodular}, and multi-armed bandit problems~\cite{bandits,bandits2}.

\label{apx:related:work}

\section{Revisiting Theorem~\ref{thm:limited:price:change:optimality}: Intuition behind Pricing Design and Proof for \oksrisk with $\delta = 1$}
\label{sec:expected}

In this section, we revisit the risk-neutral setting where $\delta = 1$.
We provide a detailed discussion of the pricing design introduced in Theorem~\ref{thm:limited:price:change:optimality}, explain the intuition behind this construction, and present a complete proof of the theorem.

\subsection{Intuition Behind the Design in Theorem~\ref{thm:limited:price:change:optimality}}
\label{sec:intuition:netural}
To elaborate on the intuition behind the pricing design in Theorem~\ref{thm:limited:price:change:optimality}, consider the fractional relaxation of \oksrisk, where the integrality constraint $x_t \in \{0,1\}$ is relaxed to $\hat x_t \in [0,1]$.  In this relaxation, the optimal online algorithm proposed by~\cite{knpasackFractional2008} determines the fractional allocation~$\hat{x}_t$ for buyer~$t$, given their value~$v_t$, according to the following utility maximization rule:
\begin{align}
\label{eq:phi-function-definition-neutral}
\hat{x}_t
= \arg\max_{x \in [0,1]}
\Biggl\{
v_t x
- k \int_{\hat{y}_t/k}^{(\hat{y}_t + x)/k} \phi(\eta)\, d\eta
\Biggr\},
\quad \text{where } \phi(x) =
\begin{cases}
L, & \text{if } x \in \bigl[0, \tfrac{1}{\alpha}\bigr), \\
L \cdot \exp(\alpha x - 1), & \text{if } x \in \bigl[\tfrac{1}{\alpha}, 1\bigr].
\end{cases}
\end{align}
Here, $\hat y_t := \sum_{s < t} \hat x_s$ denotes the algorithm’s cumulative (fractional) allocation just before buyer $t$ arrives.
Intuitively, $\phi(\eta)$ represents the marginal price at the normalized utilization level $\eta \in [0,1]$, and the integral term captures the total cost of allocating an additional $x$ units when the current utilization is $\hat y_t/k$.

To see how pricing design arise naturally from this perspective, partition the inventory into two quotas, $q_1$ and $q_2$, such that $q_1 + q_2 = k$. There exists an equivalent form of the above maximization expression that yields the same fractional decision. By an appropriate change of variables, one can verify that $\hat x_t$ can also be expressed as follows:
\begin{align*}
    \hat x_t & = \arg\max_{x \in [0,1]} \bigg\{ v_t x  -  q_1 \cdot \int_{\min\{\hat y_t/q_1,\,1\}}^{\min\{1,(\hat y_t + x)/q_1\}} \phi\big((q_1/k)\,\eta\big)\, d\eta  -  \\
    &\qquad \mathbf{1}\{\hat y_t + x > q_1\} \cdot q_2 \cdot \int_{\max\{0,\frac{\hat y_t-q_1}{q_2}\}}^{\min\{1,(\hat y_t + x - q_1)_+/q_2\}} \phi\big((q_1/k) + (q_2/k)\,\eta\big)\, d\eta \bigg\}.
    \end{align*}
The first integral prices the portion of the allocation that lies within the first $q_1$ units, using the curve $\eta \mapsto \phi\big((q_1/k)\eta\big)$. Once the cumulative allocation reaches $q_1$ (i.e., when $\hat y_t \ge q_1$), this term becomes inactive. The second integral prices the spillover into the next $q_2$ units using a shifted curve $\eta \mapsto \phi\big((q_1/k) + (q_2/k)\eta\big)$, which only contributes when the decision interval $(\hat y_t, \hat y_t + x)$ crosses the boundary at $q_1$. Moreover, the pricing of each segment of the items is normalized so that the fraction of that segment lies within the unit interval $[0,1]$. For example, for the first $q_1$ units, the integration over the curve $\eta \mapsto \phi\big((q_1/k)\eta\big)$ is normalized to the range $[0,1]$. Consequently, the design of the pricing functions in Theorem~\ref{thm:limited:price:change:optimality} is inspired by this structure—where the pricing for the first $q_1$ units and the subsequent $q_2$ units follows directly from the pricing functions used in utility-maximization formulation described above.

\subsection{Proof of Theorem~\ref{thm:limited:price:change:optimality}}

Before detailing the proof, we first introduce a set of notations that will be utilized throughout the remainder of this paper.

\paragraph{Notation.}
Let $y^{(r)}_T$ denote the number of units sold when the random seed $R$ realizes to $r \in [0,1]$. This quantity is deterministic given some instance $I$ of the \oksrisk problem as input. For simplicity we drop the subscript $T$ and write $y^{(r)}$.
Let $y^* \coloneqq \max_{r \in [0,1]} y^{(r)}$ be the maximum over realizations of the random seed $R$, and let
\begin{align*}
i^* \coloneqq 
\begin{cases}
0, & \text{if } y^* < q_1, \\
\max \left\{ i \in \{1,2,\dots,\Delta+1\} \ \middle|\  y^* \ge \sum_{l=1}^{i} q_l \right\}, & \text{otherwise},
\end{cases}
\end{align*}
which denotes the highest price level such that, once the utilization of algorithm is equal to $y^*$, 
the algorithm fully allocates all reserved units from price levels $1$ through~$i^*$. 
Let $r^* \in [0,1]$ denote the largest value of the random seed~$R$ under which the algorithm fully utilizes the reserved units from the first up to the $i^*$-th price level, i.e., 
$r^* = \max\{\, r \in [0,1] : y^{(r)} \ge \sum_{i=1}^{i^*} q_i \,\}$. 
Furthermore, let us define the function $\phi_i^*:[L,U]\rightarrow [0,1]$ where $\phi_i^*(v) = \sup\{x\in[0,1] | \phi_i(x) \leq v\}$ is defined as the general inverse of the $\phi_i$ function.

\paragraph{Proof Overview.}
We first establish two key structural properties of Algorithm~\ref{alg:corr:PPM:limited:price:change}. 
The first is a monotonicity property concerning the number of items sold, and the second provides a lower bound on the utilization level of the algorithm across all realizations of the random seed~$R$.
We then formulate a dual linear program whose optimal objective value serves as an upper bound on the offline optimal welfare. 
Next, following a randomized online primal-dual (R-OPD) approach, we define dual variable updates as functions of the realized random seed and construct a candidate dual solution. 
We then show that the dual constraints are $\alpha^{\star}$-feasible in expectation under these updates and that each buyer contributes at most $v_t$ in expectation to the dual objective. 
By weak duality, this implies that the algorithm achieves at least a $1/\alpha^{\star}$ fraction of the offline optimal value, thereby establishing its $\alpha^{\star}$-competitiveness.

We now proceed with the detailed proof of Theorem~\ref{thm:limited:price:change:optimality}.

\begin{proof}
Let us fix an input instance $I$ of the \oksrisk problem and prove Algorithm~\ref{alg:corr:PPM:limited:price:change} is $\alpha^{\star} = 1+\ln(U/L)$-competitive on this instance of the problem. Although Algorithm~\ref{alg:corr:PPM:limited:price:change} is randomized, its correlated pricing scheme ensures that the number of units sold does not vary significantly across different sample paths. 
In other words, the total number of sold units remains close to the maximum $\sum_{i=1}^{i^*} q_i$ achieved across all realizations of the random seed. 
This observation is made precise through the following two properties.

\paragraph{Monotonic Utilization.} 
We establish a monotonicity property for the number of units sold by Algorithm~\ref{alg:corr:PPM:limited:price:change} as a function of the realized random seed~$R$. 
Specifically, we show that the utilization of the algorithm by the arrival of the $t$-th buyer, denoted $y^{(r)}_{t}$, is \emph{nonincreasing} in~$r$. 
This monotonicity arises from the fact that the pricing profiles posted by the algorithm become progressively more aggressive as the random seed~$R$ increases from~$0$ to~$1$.

\begin{lemma}
\label{lem:monotonicity:1}
For any $r_1, r_2 \in [0,1]$ with $r_1 \le r_2$, and any buyer $t \in [T]$, we have $y_t^{(r_1)} \ge y_t^{(r_2)}$.
\end{lemma}

The proof of above lemma can be found in Appendix~\ref{apx:lem:monotonicity:1}.
Thus, following the above lemma, we can see that $y^{(r)} \ge \sum_{i=1}^{i^*} q_i$ for all values of $ r \in [0,r^*]$. 

\paragraph{Lower Bound on the Utilization $y^{(r)}$.} 
Given the constraint on the reservation vector $q_1 \le q_2 \le \dots \le q_{\Delta+1}$, and noting that the level-wise pricing functions are nondecreasing across levels—i.e., $\phi_j(\cdot) \le \phi_{j+1}(\cdot)$ by design—we obtain the following lower bound on the total utilization~$y^{(r)}$.

\begin{lemma}
\label{lem:lower-bound:yr}
If $i^* \ge 2$, then for all $r \in [0,1]$, we have
$
y^{(r)} \ge \sum_{i=1}^{i^*-1} q_i,
$
under the reservation-vector constraint $q_1 \le q_2 \le \dots \le q_{\Delta+1}$.
\end{lemma}

The proof of above lemma can be found in Appendix~\ref{apx:lem:lower-bound:yr}. In the subsequent proof, we utilize the two aforementioned properties, in conjunction with the $\text{R-OPD}$ framework, to establish the optimality of the competitive ratio for the pricing design detailed in Theorem~\ref{thm:limited:price:change:optimality}.

Moving forward, we assume that $\phi_{i^*+1}(r^*) > L$. This assumption is without loss of generality. Indeed, if for some instance $I$ we have $\phi_{i^*+1}(r^*) \le L$, then all buyers in that instance are accepted by \algname. To see this, note that the highest posted price used by \algname on instance $I$ is at most $\phi_{i^*+1}(r^*) \le L$. Since every buyer value satisfies $v_t \ge L$ by Assumption~\ref{ass:bounded-value}, every buyer accepts the posted price. Therefore, in this case, \algname accepts all buyers in instance $I$.

Consider the following dual linear program (LP), which upper bounds the offline optimum:
\begin{align}
\label{eq:dual:oks:expected}
\min_{{u_t}, {\lambda_j}} \quad \sum_{t \in [T]} u_t + \sum_{j=1}^{\Delta+1} \lambda_j \cdot  q_j
\qquad
\text{s.t.} \quad
v_t \le u_t + \frac{1}{k} \sum_{j=1}^{\Delta+1} \lambda_j \cdot q_j, \quad \forall t \in [T].
\end{align}
It can be verified that the optimal objective value of this LP provides an upper bound on the performance of the offline clairvoyant algorithm.
Following the economic interpretation of the randomized primal-dual framework presented in \cite{eden2021economics}, we can interpret the variable $\lambda_j$ as the \emph{price} associated with the $j$-th set of reserved units at the $j$-th price level, and the variable $u_t$ as the \emph{utility} of buyer~$t$ resulting from participating in the pricing scheme implied by the primal-dual construction.

Following the R-OPD framework, we construct, for each realization of the random seed $R = r$, a corresponding set of dual variables ${ \lambda_j^{(r)}, u_t^{(r)} }$.
The final dual variables are then defined as their expectations over the random seed such that $\lambda_j = \mathbb{E}_R[\lambda_j^{(R)}]$ and $
u_t = \mathbb{E}_R[u_t^{(R)}],
$ where the expectation is taken with respect to the random seed~$R$.

Initialize all dual variables to zero. Then for a realization of the random seed $R = r$, let us update the dual variables $\lambda_j^{(r)}$ as follows:
\begin{align}
\label{eq:limit:dual:update:lambda}
\lambda_i^{(r)} =
\begin{cases}
\phi_i(r), & i \in \{1,2,\dots,i^*-1\}, \\[2pt]
\phi_i(r), & i = i^*,\, r \in [0,\, r^*], \\[2pt]
0,         & \text{otherwise}.
\end{cases}
\end{align}
Furthermore, if buyer $t$ receives one unit from the $i$-th price level, set
\begin{align}
\label{eq:limit:dual:update:u}
u_t^{(r)} = 
\begin{cases}
v_t - \phi_i(r), & \text{if } i < i^* \text{ or } (i = i^* \text{ and } r \le r^*), \\[2pt]
v_t,             & \text{otherwise}.
\end{cases}
\end{align}

The update in Eq.~\eqref{eq:limit:dual:update:lambda} mirrors the posted price at level~$i$ whenever, under realization~$r$, the reserved units at that level are fully utilized. By the lower bound established in Lemma~\ref{lem:lower-bound:yr}, which guarantees that $y^{(r)} \ge \sum_{l=1}^{i^*-1} q_l$ for all $r \in [0,1]$, and based on the structural monotonicity proved in Lemma~\ref{lem:monotonicity:1}, the reserved units for the $i^*$-th level are fully utilized for all $r \in [0, r^*]$. 
The update in Eq.~\eqref{eq:limit:dual:update:u} sets $u_t^{(r)}$ to the buyer’s utility when buyer~$t$ is allocated a unit from a fully utilized price level (that is, a level $i < i^*$, or level~$i^*$ when $r \le r^*$), according to the price posted at that level and the buyer’s value. Taking expectations of these per-realization dual variables over~$R$ produces the final dual solution $(\{u_t\}, \{\lambda_j\})$ used in the R-OPD analysis.

We next show that the dual objective value of the solution obtained from the above updates equals the expected performance of  
Algorithm~\ref{alg:corr:PPM:limited:price:change} on instance~$I$.  
It suffices to prove that, under any realization $R = r$, $ \sum_{t \in [T]} u^{(r)}_t  +  \sum_{j=1}^{\Delta+1} \lambda^{(r)}_j \, q_j = \alg^{(r)}(I)$. Fix a realization $R = r$.  
Let $B^{(r)}$ denote the set of buyers who are allocated a unit under this realization.  
For each price level~$j$, let $B_j^{(r)} \subseteq B^{(r)}$ denote the subset of buyers served from the $j$-th level’s reserved units, so that $\lvert B_j^{(r)} \rvert$ represents the number of $j$-level units actually sold under~$r$.  

From Lemma~\ref{lem:lower-bound:yr}, we know that the algorithm fully utilizes the first $(i^*-1)$ levels of reserved units for all realizations of~$R$. Furthermore, following the definition of~$r^*$ and the monotonicity established in Lemma~\ref{lem:monotonicity:1}, all reserved units up to the $i^*$-th level are also fully utilized for realized values of the random seed in the range~$[0, r^*]$.  
Thus, we have  $ \lvert B_j^{(r)} \rvert = q_j, \ \forall j \in \{1, \dots, i^*-1\}, 
$ and $ 
\lvert B_{i^*}^{(r)} \rvert = q_{i^*}, \text{ if } r \le r^*.
$

Since the dual variable $u^{(r)}_t$ is updated only for buyers who receive an item under realization~$r$, 
as specified by Eq.~\eqref{eq:limit:dual:update:u}, we can express the total contribution from the $u_t^{(r)}$-variables as
\begin{align*}
\sum_{t \in [T]} u^{(r)}_t 
&= \sum_{j=1}^{\Delta+1} \sum_{t \in B^{(r)}_{j}} v_t \\
&= \sum_{j=1}^{i^*-1} \sum_{t \in B_j^{(r)}} (v_t - \phi_j(r))
  + \mathbf{1}\{r \le r^*\} \sum_{t \in B_{i^*}^{(r)}} (v_t - \phi_{i^*}(r))
  + \mathbf{1}\{r > r^*\} \sum_{t \in B_{i^*}^{(r)}} v_t
 \\
 & \quad + \sum_{j=i^*+1}^{\Delta+1} \sum_{t \in B_{j}^{(r)}} v_t, \\
&=  \sum_{j=1}^{\Delta+1} \sum_{t \in B_{j}^{(r)}} v_t 
   - \sum_{j=1}^{i^*-1} q_j \phi_j(r)
   - \mathbf{1}\{r \le r^*\} \, q_{i^*} \phi_{i^*}(r).
\end{align*}
Here, the second equality follows from the dual update rule in Eq.~\eqref{eq:limit:dual:update:u}, 
by decomposing the summation over all price levels and separating the case of the $i^*$-th price level based on the realized value of the random seed. 
The third equality follows from the sizes of the sets $B^{(r)}_{j}$ established above.  

Using the dual update rule for the $\lambda_j^{(r)}$ variables from Eq.~\eqref{eq:limit:dual:update:lambda}, 
the total dual contribution associated with these variables is given by
\begin{align*}
\sum_{j=1}^{\Delta+1} \lambda^{(r)}_j \, q_j
=  \sum_{j=1}^{i^*-1} q_j \phi_j(r)
  + \mathbf{1}\{r \le r^*\} \, q_{i^*} \phi_{i^*}(r).
\end{align*}
Combining the above two expressions for $\sum_t u_t^{(r)}$ and $\sum_j \lambda_j^{(r)} q_j$, we obtain
\begin{align*}
\sum_{t \in [T]} u^{(r)}_t  +  \sum_{j=1}^{\Delta+1} \lambda^{(r)}_j \, q_j
&= \sum_{i=1}^{\Delta+1} \sum_{t \in B_i^{(r)}} v_t
= \alg^{(r)}(I),
\end{align*}
which shows that the dual objective under realization~$r$ exactly equals the realized welfare of the algorithm.
Taking expectations with respect to~$R$ on both sides of the equality above yields the desired identity.

For every buyer $t \in [T]$, we show that the dual constraint in Eq.~\eqref{eq:dual:oks:expected} corresponding to this buyer is $\alpha^{\star}$-feasible in expectation; that is, $
\mathbb{E}_R \left[\, u^{(R)}_t + \frac{1}{k} \sum_{j=1}^{\Delta+1} \lambda^{(R)}_j \, q_j \,\right]
\ge \frac{v_t}{\alpha^{\star}}.
$
Establishing this inequality completes the R-OPD analysis and proves the $\alpha^{\star}$-competitiveness of
Algorithm~\ref{alg:corr:PPM:limited:price:change} based on the framework used in \cite{devanur2013}.
Consider a buyer $t$ with value $v_t$ such that, for some $i \in [\Delta+1]$, we have 
$\phi_i(0) \le v_t \le \phi_i(1)$. 
Depending on the value of $i$ and $v_t$, we prove the feasibility of the dual constraint under several different scenarios in what follows.

\textbf{Case I:} Either $i \le i^* - 1$, or $i = i^*$ and $\phi_i^*(v_t) \le r^*$. From the dual update rule defined in Eq.~\eqref{eq:limit:dual:update:lambda}, and noting that for all realizations of $R \in [0,1]$, the reserved units corresponding to the first $i^* - 1$ price levels are fully sold, and for $R \in [0, r^*]$, the reserved units at the $i^*$-th price level are also fully sold, we have
\begin{align*}
    \mathbb{E} \left[u_t + \frac{\sum_{l=1}^{\Delta+1} \lambda_l q_l}{k}\right]
    &\ge \sum_{l=1}^{i^*-1} \frac{q_l}{k} \int_{0}^{1} \phi_l(\eta)\, d\eta
    + \frac{q_{i^*}}{k} \int_{0}^{r^*} \phi_{i^*}(\eta)\, d\eta \\
    &= \sum_{l=1}^{i^*-1} \int_{\sum_{m=1}^{l-1} \frac{q_m}{k}}^{\sum_{m=1}^{l} \frac{q_m}{k}} \phi(\eta)\, d\eta
    + \int_{\sum_{l=1}^{i^*-1} \frac{q_l}{k}}^{\sum_{l=1}^{i^*-1} \frac{q_l}{k} + \frac{q_{i^*}}{k} r^*} \phi(\eta)\, d\eta \\
    &= \int_{0}^{\sum_{l=1}^{i^*-1} \frac{q_l}{k} + \frac{q_{i^*}}{k} r^*} \phi(\eta)\, d\eta \\
    &= \frac{\phi \left(\sum_{l=1}^{i^*-1} \frac{q_l}{k} + \frac{q_{i^*}}{k} r^*\right)}{\alpha^{\star}} \\
    &\ge \frac{v_t}{\alpha^{\star}}.
\end{align*}

The first inequality follows from the dual update rule in Eq.~\eqref{eq:limit:dual:update:lambda}.  
The following equalities holds by the construction of the pricing functions $\phi_j$, as described in Theorem~\ref{thm:limited:price:change:optimality} (see Eq.~\eqref{eq:phi:design:limited:price}) and the definition of $\phi$ function according to Eq.~\eqref{eq:phi-function-definition-neutral}.


Let us now continue the proof of the feasibility of the dual constraints for each buyer~$t$ by considering the remaining cases.

\textbf{Case II}: $i = i^*$ and $\phi_i^{-1}(v_t) > r^*$. In this case, we consider the following two subcases.

\textit{Subcase I}: For some realized value of $R \in [r^*,  \phi_i^{-1}(v_t))$, buyer $t$ is allocated a unit from among the first $\sum_{l=1}^{i-1} q_l$ reserved units.
We show that, with probability one, buyer $t$ is always allocated a unit. Consider such an $r' \in [r^*,  \phi_i^{-1}(v_t))$ where buyer $t$ receives a unit from the first $\sum_{l=1}^{i-1} q_l$ reserved units when $R = r'$. 
By Lemma~\ref{lem:lower-bound:yr}, for all $r \in [r',  1]$, the utilization level at the arrival of buyer $t$ satisfies $y_t^{(r)} \le y_t^{(r')}$.
Therefore, for any realization $R \in [r',  1]$, at least one unit from the first $\sum_{l=1}^{i-1} q_l$ units is always available for allocation to buyer $t$.
Moreover, for any realization of $R$ within $[0,  r')$, the utilization level $y_{t-1}^{(R)}$ is strictly less than $\sum_{l=1}^{i} q_l$. 
Otherwise, following an argument similar to that in Lemma~\ref{lem:monotonicity:1}, it would imply that $y_{t-1}^{(r')} \ge \sum_{l=1}^{i-1} q_l$, contradicting the assumption that under $r'$, buyer $t$ is allocated a unit from the first $\sum_{l=1}^{i-1} q_l$ units. Hence, from the above analysis, buyer $t$ is allocated a unit with probability one. 

Using the dual update rules in Eq.~\eqref{eq:limit:dual:update:u} and Eq.~\eqref{eq:limit:dual:update:lambda}, we obtain
\begin{align*}
    \mathbb{E} \left[u_t + \frac{1}{k} \sum_{l=1}^{\Delta+1} \lambda_l q_l \right]
    &\ge \sum_{l=1}^{i^*-1} \frac{q_l}{k} \int_{0}^{1} \phi_l(\eta)  d\eta
    + \frac{q_{i^*}}{k} \int_{0}^{r^*} \phi_{i^*}(\eta)  d\eta \\
    &\quad + v_t - \int_{0}^{r^*} \phi_{i^*}(\eta)  d\eta - \int_{r^*}^{1} \phi_{i^*-1}(\eta)  d\eta \\[1ex]
    &= \sum_{l=1}^{i^*-1} \int_{\sum_{m=1}^{l-1} \frac{q_m}{k}}^{\sum_{m=1}^{l} \frac{q_m}{k}} \phi(\eta)  d\eta
    + \int_{\sum_{l=1}^{i^*-1} \frac{q_l}{k}}^{\sum_{l=1}^{i^*-1} \frac{q_l}{k} + \frac{q_{i^*}}{k} r^*} \phi(\eta)  d\eta \\
    &\quad + v_t - \int_{0}^{r^*} \phi_{i^*}(\eta)  d\eta - \int_{r^*}^{1} \phi_{i^*-1}(\eta)  d\eta \\[1ex]
    &\ge \int_{0}^{\sum_{l=1}^{i^*-1} \frac{q_l}{k} + \frac{q_{i^*}}{k} r^*} \phi(\eta)  d\eta
    + \int_{\sum_{l=1}^{i^*-1} \frac{q_l}{k} + \frac{q_{i^*}}{k} r^*}^{\sum_{l=1}^{i^*-1} \frac{q_l}{k} + \frac{q_{i^*}}{k} \phi_{i^*}^{-1}(v_t)} \phi(\eta)  d\eta \\[1ex]
    &= \frac{1}{\alphathm}   \phi \left(\sum_{l=1}^{i^*-1} \frac{q_l}{k} + \frac{q_{i^*}}{k} r^*\right)
    \ge \frac{v_t}{\alphathm}.
\end{align*}

In the above derivation:
The first two terms follow from the dual update rule in Eq.~\eqref{eq:limit:dual:update:lambda}, together with the facts that $y_T^{(r)} \ge \sum_{l=1}^{i^*-1} q_l$ for all $r \in [0,1]$, and $y_T^{(r)} \ge \sum_{l=1}^{i^*} q_l$ for all $r \in [0, r^*]$. The last three terms on the left-hand side of the first inequality follow from the dual update for buyer $t$ in Eq.~\eqref{eq:limit:dual:update:u}. The second inequality follows from the construction of the functions $\phi_j$ in Theorem~\ref{thm:limited:price:change:optimality}. The final inequality is implied by the following lemma.

\begin{lemma}
\label{lemma:limit:dual:feasibility:1}
For any buyer $t$ with value $v_t \in [\phi_{i^*}(0),  \phi_{i^*}(1)]$, 
such that $\phi_{i^*}^{-1}(v_t) \ge r^*$, the following inequality holds:
\begin{align*}
v_t
- \int_{r^*}^{1} \phi_{i^*-1}(\eta)  d\eta
- \int_{0}^{r^*} \phi_{i^*}(\eta)  d\eta
\ge
\int_{\displaystyle \sum_{l=1}^{i^*-1} \frac{q_l}{k} + \frac{q_{i^*}}{k}  r^*}^{\displaystyle \sum_{l=1}^{i^*-1} \frac{q_l}{k} + \frac{q_{i^*}}{k}  \phi_{i^*}^{-1}(v_t)}
\phi(\eta)  d\eta .
\end{align*}
\end{lemma}

The proof of the above lemma is provided in Appendix~\ref{apx:lemma:limit:dual:feasibility:1}. The argument follows from the construction of the pricing functions described in Theorem~\ref{thm:limited:price:change:optimality}. This concludes the proof of dual constraint feasibility for buyer~$t$, whose valuation satisfies the conditions of Case~2, Subcase~1.

\textit{Subcase 2:} For no value of $r \in [r^*,  \phi_i^{-1}(v_t))$ is a unit from the first $\sum_{l=1}^{i-1} q_l$ reserved units allocated to buyer~$t$. In this case, whenever $r \in [r^*,  w]$, buyer~$t$ is allocated a unit from the reserved units corresponding to the $i^*$-th price level. 
From the dual updates, we have
\begin{align*}
    \mathbb{E} \left[u_t + \frac{1}{k} \sum_{l=1}^{\Delta+1} \lambda_l q_l \right]
    &\ge \sum_{l=1}^{i^*-1} \frac{q_l}{k} \int_{0}^{1} \phi_l(\eta)  d\eta
      + \frac{q_{i^*}}{k} \int_{0}^{r^*} \phi_{i^*}(\eta)  d\eta
      + v_t (w - r^*) \\[1ex]
    &= \sum_{l=1}^{i^*-1} 
       \int_{\sum_{m=1}^{l-1} \frac{q_m}{k}}^{\sum_{m=1}^{l} \frac{q_m}{k}} 
       \phi(\eta)  d\eta
       + \int_{\sum_{l=1}^{i^*-1} \frac{q_l}{k}}^{\sum_{l=1}^{i^*-1} \frac{q_l}{k} + \frac{q_{i^*}}{k} r^*} 
         \phi(\eta)  d\eta
       + v_t (w - r^*) \\[1ex]
    &\ge \int_{0}^{\sum_{l=1}^{i^*-1} \frac{q_l}{k} + \frac{q_{i^*}}{k} r^*} 
           \phi(\eta)  d\eta
       + \int_{\sum_{l=1}^{i^*-1} \frac{q_l}{k} + \frac{q_{i^*}}{k} r^*}^{\sum_{l=1}^{i^*-1} \frac{q_l}{k} + \frac{q_{i^*}}{k} w} 
         \phi(\eta)  d\eta \\[1ex]
    &= \frac{1}{\alphathm}  
       \phi \left(\sum_{l=1}^{i^*-1} \frac{q_l}{k} + \frac{q_{i^*}}{k} w\right)
      = \frac{v_t}{\alphathm}.
\end{align*}

\textbf{Case III:} $i = i^* + 1$ and $\phi_i^{-1}(v_t) < r^*$.  
Let $w = \phi_i^{-1}(v_t)$.

\textit{Subcase 1:} For some $r \in [0,  w]$, one of the first $\sum_{l=1}^{i^*} q_l$ reserved units is allocated to buyer~$t$. Then, following the same reasoning as in Subcase~1 of Case~2, a unit of the item is allocated to buyer~$t$ with probability one.  
Hence, we have
\begin{align*}
    \mathbb{E} \left[u_t + \frac{1}{k} \sum_{l=1}^{\Delta+1} \lambda_l q_l \right]
    &\ge \sum_{l=1}^{i^*-1} \frac{q_l}{k} \int_{0}^{1} \phi_l(\eta)  d\eta
      + \frac{q_{i^*}}{k} \int_{0}^{r^*} \phi_{i^*}(\eta)  d\eta \\
    &\quad + v_t - \int_{0}^{r^*} \phi_{i^*}(\eta)  d\eta
      - \int_{r^*}^{1} \phi_{i^*-1}(\eta)  d\eta \\[1ex]
    &= \sum_{l=1}^{i^*-1} 
       \int_{\sum_{m=1}^{l-1} \frac{q_m}{k}}^{\sum_{m=1}^{l} \frac{q_m}{k}} 
       \phi(\eta)  d\eta
       + \int_{\sum_{l=1}^{i^*-1} \frac{q_l}{k}}^{\sum_{l=1}^{i^*-1} \frac{q_l}{k} + \frac{q_{i^*}}{k} r^*} 
         \phi(\eta)  d\eta \\
    &\quad + v_t - \int_{0}^{r^*} \phi_{i^*}(\eta)  d\eta
      - \int_{r^*}^{1} \phi_{i^*-1}(\eta)  d\eta \\[1ex]
    &\ge \int_{0}^{\sum_{l=1}^{i^*-1} \frac{q_l}{k} + \frac{q_{i^*}}{k} r^*} \phi(\eta)  d\eta
      + \int_{\sum_{l=1}^{i^*-1} \frac{q_l}{k} + \frac{q_{i^*}}{k} r^*}^{\sum_{l=1}^{i^*-1} \frac{q_l}{k} + \frac{q_{i^*}}{k} \phi_{i^*}^{-1}(v_t)} 
        \phi(\eta)  d\eta \\[1ex]
    &= \frac{1}{\alphathm} 
       \phi \left(\sum_{l=1}^{i^*-1} \frac{q_l}{k}
       + \frac{q_{i^*}}{k} r^*\right)
      \ge \frac{v_t}{\alphathm}.
\end{align*}

The first inequality and the second equality follow by the same reasoning as in the preceding cases.  
Furthermore, the second inequality follows from the lemma stated below.

\begin{lemma}
\label{lemma:limit:dual:feasibility:2}
For any buyer $t$ with value $v_t \in [\phi_{i^*+1}(0),  \phi_{i^*+1}(1)]$, 
such that $\phi_{i^*+1}^{-1}(v_t) \leq r^*$, the following inequality holds:
\begin{align*}
v_t 
- \int_{\eta = r^*}^{1} \phi_{i^*-1}(\eta)   d\eta
- \int_{\eta = 0}^{r^*} \phi_{i^*}(\eta)   d\eta
\ge 
\int_{\eta = \sum_{l=1}^{i^*-1} \frac{q_l}{k} + \frac{q_{i^*}}{k} \cdot r^*}^{\sum_{l=1}^{i^*} \frac{q_l}{k} + \frac{q_{i^*+1}}{k} \cdot w} \phi(\eta)   d\eta.
\end{align*}
\end{lemma}

The proof of above lemma can be found in Appendix~\ref{apx:lemma:limit:dual:feasibility:2} which follows from the pricing function design in Theorem~\ref{thm:limited:price:change:optimality}.

\textit{Subcase 2: }
For no value of $r \in [0,  w]$ is any of the first $\sum_{l=1}^{i^*} q_l$ reserved units allocated to buyer~$t$.  
Then, it must be that for $r \in [r^*,  1]$, a unit from the reserved units at the $i^*$-th price level is allocated to buyer~$t$, and for $r \in [0,  w]$, a unit from the $(i^*+1)$-th price level is allocated to buyer~$t$.  
Following the dual updates described above, we obtain
\begin{align*}
    \mathbb{E} \left[u_t + \frac{1}{k} \sum_{l=1}^{\Delta+1} \lambda_l q_l \right]
    &\ge \sum_{l=1}^{i^*-1} \frac{q_l}{k} \int_{0}^{1} \phi_l(\eta)  d\eta
      + \frac{q_{i^*}}{k} \int_{0}^{r^*} \phi_{i^*}(\eta)  d\eta
      + v_t \cdot (1 - r^* + w) \\[1ex]
    &\ge \sum_{l=1}^{i^*-1} 
       \int_{\sum_{m=1}^{l-1} \frac{q_m}{k}}^{\sum_{m=1}^{l} \frac{q_m}{k}} 
       \phi(\eta)  d\eta
       + \int_{\sum_{l=1}^{i^*-1} \frac{q_l}{k}}^{\sum_{l=1}^{i^*-1} \frac{q_l}{k} + \frac{q_{i^*}}{k} r^*} 
         \phi(\eta)  d\eta
       + v_t \cdot (1 - r^* + w) \\[1ex]
    &\ge \int_{0}^{\sum_{l=1}^{i^*-1} \frac{q_l}{k} + \frac{q_{i^*}}{k} r^*} \phi(\eta)  d\eta
      + \int_{\sum_{l=1}^{i^*-1} \frac{q_l}{k} + \frac{q_{i^*}}{k} r^*}^{\sum_{l=1}^{i^*} \frac{q_l}{k} + \frac{q_{i^*+1}}{k} w} 
        \phi(\eta)  d\eta \\[1ex]
    &= \frac{1}{\alphathm} 
       \phi \left(\sum_{l=1}^{i^*} \frac{q_l}{k}
       + \frac{q_{i^*+1}}{k} w\right)
       \ge \frac{v_t}{\alphathm}.
\end{align*}

\textbf{Case IV:} Either $i = i^* + 1$ and $\phi_i^{-1}(v_t) > r^*$, or $i > i^* + 1$.  The proof for this case follows by reasoning analogous to the previous cases. Therefore, considering all four cases, the $\alphathm$-feasibility of the dual constraint corresponding to each buyer~$t$ is verified. Consequently, the $\alphathm$-competitiveness of Algorithm~\ref{alg:corr:PPM:limited:price:change} follows.

Considering all cases together, the $\alpha^{\star}$-feasibility of the dual constraint corresponding to each buyer~$t$ is thus verified. Consequently, the $\alpha^{\star}$-competitiveness of Algorithm~\ref{alg:corr:PPM:limited:price:change} follows.
Based on the lower bound of $1 + \ln(U/L)$ established in prior work~\cite{sun2024static} for the attainable competitive ratio of any online algorithm for the $k$-selection problem, which is a special case of \oksrisk, \algname attains the optimal worst-case competitive ratio when $\delta = 1$ and value of $\Delta$ varies in the range $\{0,1,\dots,k-1\}$.

\end{proof}

\subsection{Proof of Lemma~\ref{lem:monotonicity:1}}
\label{apx:lem:monotonicity:1}
Let $B_{t}^{(r_2)}$ be the set of buyers to whom a unit is allocated by the arrival of buyer~$t$ under realization~$r_2$.
Consider the first buyer~$t'$ in $B_{t}^{(r_2)}$. If $y^{(r_1)}_{t'} = 0$ at the arrival of~$t'$, then $t'$ accepts the posted price
because $\phi_{1}(r_1) \le \phi_1(r_2) \le v_{t'}$, and thus $y^{(r_1)}_{t'}$ increases to~$1$.
Proceeding inductively, suppose that upon the arrival of the $\ell$-th buyer in $B_t^{(r_2)}$ we have
$y^{(r_1)}_{t'} = \ell - 1 < y^{(r_2)}_{t'} = \ell$. Then $v_{t'}$ is at least the posted price for the $\ell$-th unit under realization~$R = r_2$.
Since the posted prices are lower under~$r_1$ than under~$r_2$, following from the nondecreasing property of the pricing functions, buyer~$t'$ also accepts under~$r_1$, implying
$y^{(r_1)}_{t'} = \ell$. Hence, $y^{(r_1)}_{t} \ge y^{(r_2)}_{t}$.

\subsection{Proof of Lemma~\ref{lem:lower-bound:yr}}
\label{apx:lem:lower-bound:yr}

For each $i$, let $B^*_i$ be the set of buyers who, under realization $r^*$, were allocated one of the
$q_i$ reserved units at the $i$-th price level. By definition, for each $i \in [i^*]$, $|B^*_i| \ge q_i$.
From the pricing design in Theorem~\ref{thm:limited:price:change:optimality}, for all $t \in B^*_i$ we have
$v_t \ge \phi_{i-1}(1)$.
Let $t'$ be the first buyer in $B^*_2$. Fix any $r \in [0,1]$ and suppose $y^{(r)}_{t'} < q_1$ at the arrival
of $t'$. Since $|B^*_2| = q_2 \ge q_1$, and each buyer in $B^*_2$ has value at least $\phi_1(1)$ while the posted
price for the first $q_1$ units is $\phi_1(R) \le \phi_1(1)$ (as $\phi_1$ is increasing), every buyer in $B^*_2$
accepts the price for those first $q_1$ units. Thus these $q_1$ units are fully sold by the arrival of the last
buyer in $B^*_2$.
Repeating the same argument inductively for all $i \le i^*$ yields, for any $r \in [0,1]$,
$y^{(r)} \ge \sum_{i=1}^{i^*-1} q_i$, as claimed.

\subsection{Proof of Lemma~\ref{lemma:limit:dual:feasibility:1} }
\label{apx:lemma:limit:dual:feasibility:1}
Set $ A = \sum_{l=1}^{i^{*}-1} \frac{q_l}{k} $ and $
w = \phi_{i^{*}}^{-1}(v_t) $. Then
\begin{align*}
v_t
&- \int_{r^*}^{1} \phi_{i^*-1}(\eta) d\eta
- \int_{0}^{r^*} \phi_{i^*}(\eta) d\eta \\
&= \phi \left(A + w \frac{q_{i^*}}{k}\right)
   - \frac{k}{\alpha q_{i^*-1}}
     \Big(\phi(A) - \phi \big(A - (1-r^*) \frac{q_{i^*-1}}{k}\big)\Big)
   - \frac{k}{\alpha q_{i^*}}
     \Big(\phi \big(A + r^* \frac{q_{i^*}}{k}\big) - \phi(A)\Big),
\end{align*}
where the first equality follows from the definition of $\phi_j$ in Eq.~\eqref{eq:phi:design:limited:price}.

To prove the lemma it suffices to show
\begin{align*}
& \phi \left(A + w   \frac{q_{i^*}}{k}\right)
- \frac{k}{\alpha q_{i^*-1}}
  \Big(\phi(A) - \phi \big(A - (1-r^*) \frac{q_{i^*-1}}{k}\big)\Big)  \\
& \hspace{3mm}  - \frac{k}{\alpha q_{i^*}}
  \Big(\phi \big(A + r^* \frac{q_{i^*}}{k}\big) - \phi(A)\Big)
\ge
\frac{1}{\alpha} \left(\phi \left(A + w   \frac{q_{i^*}}{k}\right) 
- \phi \left(A + r^*  \frac{q_{i^*}}{k} \right)\right),
\end{align*}
where the right-hand side follows from the definition of the $\phi$ function according to Eq.~\eqref{eq:phi-function-definition-neutral}. Define
\begin{align*}
F = {}&
\phi \left(A + w \frac{q_{i^*}}{k}\right)
- \frac{k}{\alpha q_{i^*-1}}
  \Big(\phi(A) - \phi \big(A - (1-r^*) \frac{q_{i^*-1}}{k}\big)\Big)
- \frac{k}{\alpha q_{i^*}}
  \Big(\phi \big(A + r^* \frac{q_{i^*}}{k}\big) - \phi(A)\Big) \\
&\quad
- \frac{1}{\alpha} \left(\phi \left(A + w \frac{q_{i^*}}{k}\right) 
- \phi \left(A + r^* \frac{q_{i^*}}{k}\right)\right).
\end{align*}
We show $F\ge 0$ for all admissible instance-dependent parameters.

\begin{align*}
F
&= \Big(1-\frac{1}{\alpha}\Big) \phi \left(A + w \frac{q_{i^*}}{k}\right)
   + \frac{1}{\alpha} \phi \left(A + r^* \frac{q_{i^*}}{k}\right) \\
&\qquad
   - \frac{k}{\alpha q_{i^*}}
     \Big(\phi \big(A + r^* \frac{q_{i^*}}{k}\big) - \phi(A)\Big)
   - \frac{k}{\alpha q_{i^*-1}}
     \Big(\phi(A) - \phi \big(A - (1-r^*) \frac{q_{i^*-1}}{k}\big)\Big) \\
&= \phi \left(A + r^* \frac{q_{i^*}}{k}\right)
   \Bigg[
      \Big(1-\frac{1}{\alpha}\Big) e^{\alpha(w-r^*)}
      + \frac{1}{\alpha} \\
      & \qquad
      - \frac{k}{\alpha q_{i^*}}\Big(1-e^{-\alpha r^* \frac{q_{i^*}}{k}}\Big)
      - \frac{k}{\alpha q_{i^*-1}}
        \Big(e^{-\alpha r^* \frac{q_{i^*}}{k}}
           - e^{-\alpha\big(r^* \frac{q_{i^*}}{k} + (1-r^*) \frac{q_{i^*-1}}{k}\big)}\Big)
   \Bigg] \\
&\ge \phi \left(A + r^* \frac{q_{i^*}}{k}\right)
   \Bigg[
      1
      - \frac{k}{\alpha q_{i^*}}\Big(1-e^{-\alpha r^* \frac{q_{i^*}}{k}}\Big)
      - \frac{k}{\alpha q_{i^*-1}}
        \Big(e^{-\alpha r^* \frac{q_{i^*}}{k}}
           - e^{-\alpha\big(r^* \frac{q_{i^*}}{k} + (1-r^*) \frac{q_{i^*-1}}{k}\big)}\Big)
   \Bigg] \\
&\ge \phi \left(A + r^* \frac{q_{i^*}}{k}\right)
   \Bigg[
      1
      - \frac{k}{\alpha q_{i^*}}\cdot \alpha r^* \frac{q_{i^*}}{k}
      - \frac{k}{\alpha q_{i^*-1}}\cdot \alpha(1-r^*) \frac{q_{i^*-1}}{k}
   \Bigg] \\
&= \phi \left(A + r^* \frac{q_{i^*}}{k}\right) \big(1 - r^* - (1-r^*)\big)
= 0,
\end{align*}
where the first inequality follows from $w \ge r^*$ (so $e^{\alpha(w - r^*)} \ge 1$), and the second inequality uses the bound $1 - e^{-x} \le x$ for $x \ge 0$. Moreover, we can verify that 
$e^{-\alpha r^* \frac{q_{i^*}}{k}} - e^{-\alpha \left(r^* \frac{q_{i^*}}{k} + (1 - r^*) \frac{q_{i^*-1}}{k}\right)} \le \alpha(1 - r^*) \frac{q_{i^*-1}}{k}$, 
since 
\begin{align*}
    e^{-\alpha r^* \frac{q_{i^*}}{k}} - e^{-\alpha \left(r^* \frac{q_{i^*}}{k} + (1 - r^*) \frac{q_{i^*-1}}{k}\right)} = \int_{\eta = -\alpha\left(r^* \frac{q_{i^*}}{k} + (1 - r^*) \frac{q_{i^*-1}}{k}\right)}^{-\alpha r^* \frac{q_{i^*}}{k}} e^{\eta}  d\eta \le \alpha(1 - r^*) \frac{q_{i^*-1}}{k},
    \end{align*}
where the above inequality holds because the integration interval lies within $\eta \in (-\infty, 0]$, and  $e^{\eta} \le 1$. Thus, $F \ge 0$, which completes the proof of the lemma.

\subsection{Proof of Lemma~\ref{lemma:limit:dual:feasibility:2}}
\label{apx:lemma:limit:dual:feasibility:2}
Let
\begin{align*} 
A=\sum_{\ell=1}^{i^*-1}\frac{q_\ell}{k},
\qquad
w=\phi_{i^*+1}^{-1}(v_t).
\end{align*}
Then, by the definition of the pricing functions in Theorem~\ref{thm:limited:price:change:optimality},
\begin{align*}
v_t
=
\phi \left(
A+\frac{q_{i^*}}{k}
+w\frac{q_{i^*+1}}{k}
\right).
\end{align*}
As in the proof of Lemma~3, subtract the right-hand side of the desired
inequality from the left-hand side and denote the resulting expression by
$F(w)$. Then
\begin{align*}
F(w)
={}&
\phi \left(A+\frac{q_{i^*}}{k}
+w\frac{q_{i^*+1}}{k}\right)
-\frac{k}{\alpha q_{i^*-1}}
\left(
\phi(A)-\phi \left(A-(1-r^*)\frac{q_{i^*-1}}{k}\right)
\right)\\
&-\frac{k}{\alpha q_{i^*}}
\left(
\phi \left(A+r^*\frac{q_{i^*}}{k}\right)-\phi(A)
\right)\\
&-\frac{1}{\alpha}
\left(
\phi \left(A+\frac{q_{i^*}}{k}
+w\frac{q_{i^*+1}}{k}\right)
-\phi \left(A+r^*\frac{q_{i^*}}{k}\right)
\right).
\end{align*}
Equivalently,
\begin{align*}
F(w)
={}&
\left(1-\frac{1}{\alpha}\right)
\phi \left(A+\frac{q_{i^*}}{k}
+w\frac{q_{i^*+1}}{k}\right)
+\frac{1}{\alpha}
\phi \left(A+r^*\frac{q_{i^*}}{k}\right)\\
&-\frac{k}{\alpha q_{i^*-1}}
\left(
\phi(A)-\phi \left(A-(1-r^*)\frac{q_{i^*-1}}{k}\right)
\right)\\
&-\frac{k}{\alpha q_{i^*}}
\left(
\phi \left(A+r^*\frac{q_{i^*}}{k}\right)-\phi(A)
\right).
\end{align*}
Since $\alpha>1$, $\phi$ is nondecreasing, and $q_{i^*+1}>0$, the
first term is nondecreasing in $w$. Therefore,
\begin{align*}
F(w)\ge F(0).
\end{align*}
It remains to show that $F(0)\ge 0$.

Let
\begin{align*}
\beta=\alpha\frac{q_{i^*}}{k},
\qquad
\gamma=\alpha\frac{q_{i^*-1}}{k}.
\end{align*}
Using the exponential form of $\phi$ from Theorem~1, and using
$1-e^{-x}\le x$ for $x\ge 0$, we obtain
\begin{align*}
\frac{F(0)}{\phi(A)}
\ge
\left(1-\frac{1}{\alpha}\right)e^\beta
+\frac{1}{\alpha}e^{r^*\beta}
-(1-r^*)
-\frac{e^{r^*\beta}-1}{\beta}.
\end{align*}
Define
\begin{align*}
h(r)
=
\left(1-\frac{1}{\alpha}\right)e^\beta
+\frac{1}{\alpha}e^{r\beta}
-(1-r)
-\frac{e^{r\beta}-1}{\beta}.
\end{align*}
We show $h(r)\ge 0$ for all $r\in[0,1]$. Its second derivative is
\begin{align*}
h''(r)
=
-\beta\left(1-\frac{\beta}{\alpha}\right)e^{r\beta}
\le 0,
\end{align*}
because $\beta/\alpha=q_{i^*}/k\le 1$. Hence $h$ is concave, so its
minimum over $[0,1]$ is attained at an endpoint. At $r=0$,
\begin{align*}
h(0)
=
\left(1-\frac{1}{\alpha}\right)(e^\beta-1)
\ge 0.
\end{align*}
At $r=1$,
\begin{align*}
h(1)
=
e^\beta-\frac{e^\beta-1}{\beta}
=
\frac{e^\beta(\beta-1)+1}{\beta}
\ge 0,
\end{align*}
where the last inequality follows because the function
$g(\beta)=e^\beta(\beta-1)+1$ satisfies $g(0)=0$ and
$g'(\beta)=\beta e^\beta\ge 0$. Therefore $h(r^*)\ge 0$, and hence
$F(0)\ge 0$. Since $F(w)\ge F(0)$, we conclude that $F(w)\ge 0$,
which proves the lemma.

\section{Revisiting Theorem~\ref{thm:design:cvar:static}: Intuition behind Pricing Design and  the Proof for the Case of \oksrisk with $\Delta = 0$}
\label{apx:static}

In this section, we first give an intuition behind the pricing design in Theorem~\ref{thm:design:cvar:static}. To do so, we first establish a lower bound on the performance of all static pricing mechanisms.
Subsequently, we present a proposition that characterizes the pricing design for an $\alpha$-\crcvar static pricing mechanism, motivated by this lower-bound analysis, and demonstrate how the optimality result established in Theorem~\ref{thm:design:cvar:static} follows from these findings.

\paragraph{Intuition behind Pricing Design in Theorem~\ref{thm:design:cvar:static}.}
The design of the pricing function in Theorem~\ref{thm:design:cvar:static} is inspired by the approach in~\cite{sun2024static,jazi2025posted}, which derives lower bounds on the competitive ratio of all online algorithms by identifying the optimal online algorithm over a class of hard instances.
The behavior of this optimal algorithm on such instances, in turn, motivates the construction of the pricing function for our optimal static algorithm. Thus, following this approach, for some value of $\epsilon \ge 0$, consider the following hard instance~$\mathcal{I}^{(\epsilon)}$ defined as:
\begin{align*}
\mathcal{I}^{(\epsilon)} =
\big\{ &
\underbrace{L, \dots, L}_{k \text{ buyers}}, \,
\underbrace{L + \epsilon, \dots, L + \epsilon}_{k \text{ buyers}},\,
\dots,\, 
\underbrace{L + j \cdot \epsilon, \dots, L + j \cdot \epsilon}_{k \text{ buyers in stage } L + j \cdot \epsilon},
\dots, \,\\
& \qquad
\underbrace{L + \lfloor (U - L)/\epsilon \rfloor \cdot \epsilon, \dots,  L + \lfloor (U - L)/\epsilon \rfloor \cdot \epsilon}_{k \text{ buyers}}
\big\}.
\end{align*}
Let $V^{(\epsilon)} = \{L, L + \epsilon, \dots, L + \lfloor (U - L)/\epsilon \rfloor \cdot \epsilon\}$ denote the set of all possible buyer values appearing in~$\mathcal{I}^{(\epsilon)}$. 
For any $v \in V^{(\epsilon)}$, let $\mathcal{I}_{v}^{(\epsilon)}$ denote the subset of buyers in~$\mathcal{I}^{(\epsilon)}$ consisting of all arrivals up to and including the $k$ buyers with value~$v$. 
We derive the optimal static pricing algorithm on the class of instances $\{\mathcal{I}_{v}^{(\epsilon)}\}_{v \in V^{(\epsilon)}}$. 
Since the instance $\mathcal{I}_{v}^{(\epsilon)}$ is identical to $\mathcal{I}_{v'}^{(\epsilon)}$ up to the arrival of the $k$ buyers with value~$v$, for some values of $v \leq v'$ in~$V^{(\epsilon)}$, the online algorithm cannot distinguish between these instances. 
Therefore, to analyze the performance of an online algorithm on the class $\{\mathcal{I}_{v}^{(\epsilon)}\}_{v \in V^{(\epsilon)}}$, we can equivalently assume that the entire instance~$\mathcal{I}^{(\epsilon)}$ is revealed to the algorithm, although the sequence may stop at any stage. 
In other words, to achieve $\alpha$-\crcvar competitiveness, the algorithm must guarantee an expected welfare of at least $k \cdot v / \alpha$ (i.e., $\frac{1}{\alpha}$ fraction of optimal clairvoyant) by the end of the stage in~$\mathcal{I}^{(\epsilon)}$ where $k$ buyers with value~$v$ arrive. 
This requirement follows because the input sequence may terminate after any stage corresponding to the arrival of $k$ buyers with value~$v \in V^{(\epsilon)}$, thereby yielding an instance from the class $\{\mathcal{I}_{v}^{(\epsilon)}\}_{v \in V^{(\epsilon)}}$.

Let the random variable~$P$ represent the price posted by a static pricing algorithm on the given instance $\mathcal{I}^{(\epsilon)}$ . 
Since the algorithm is static, we have $p_t = P$ for all $t \in \{1,2,\dots,T\}$. 
Define the function $\psi:[L,U] \rightarrow [0,1]$ such that $\psi(x)$ denotes the probability that the algorithm posts a price less than or equal to~$x$, i.e.,
$
\Pr[P \le x] = \psi(x).
$
The following condition must hold for a fully-static pricing algorithm with pricing distribution characterized by~$\psi$ to be $\alpha$-\crcvar on the class of hard instances $\{\mathcal{I}_{v}^{(\epsilon)}\}_{v \in V^{(\epsilon)}}$ :

\begin{proposition}
\label{prop:lb:static}
For the class of hard instances $\{\mathcal{I}_{v}^{(\epsilon)}\}_{v \in V^{(\epsilon)}}$, a fully-static pricing algorithm with price distribution function~$\psi$ must satisfy the following constraints to be $\alpha$-\crcvar competitive:
\begin{align*}
& \psi(L)  \ge 1 - \delta + \delta \cdot \frac{1}{\alpha}, \\
& L  \cdot \frac{1}{\delta} \cdot \min \{ \delta - 1 + \psi(v), \psi(L) \}
+ \int_{\eta = \psi(L)}^{\max\{\psi(L), \delta - 1 + \psi(v)\}}    \frac{1}{\delta} \cdot \, \psi^*(\eta) d\eta
\ge \frac{ v}{\alpha}, \quad \forall v \in V^{(\epsilon)},
\end{align*}
where the function $\psi^*(x) = \sup\{v \in [L,U] | \psi(v) \leq x\} $.
\end{proposition}

The proof of the above proposition can be found in Appendix~\ref{apx:prop:lb:static}. 
In the above, the first inequality follows from the fact that, with probability at least $1 - \delta + \delta \cdot \tfrac{1}{\alpha}$, any online algorithm must post the price~$L$.
Otherwise, the \cvar performance of the online algorithm on instance~$I^{(\epsilon)}_{L}$ would fall below $\tfrac{k \cdot L}{\alpha}$, which is the benchmark value required for the algorithm to be $\alpha$-\crcvar.
Furthermore, the left-hand side of the second inequality represents the \cvar objective value of an online algorithm up to the end of the stage in which $k$ buyers with value~$v$ arrive, where the algorithm’s randomness in the posted price is captured by the function~$\psi$.

Without loss of generality, as $\epsilon \rightarrow 0$, we can assume that the function~$\psi$, corresponding to the optimal online algorithm on the class of hard instances $\{\mathcal{I}_{v}^{(\epsilon)}\}_{v \in V^{(\epsilon)}}$, is a continuous strictly increasing function. Obtaining the function~$\psi$ that satisfies the above inequality for all $v \in [L, U]$ with the smallest possible value of $\alpha$ 
naturally motivates a corresponding design for the pricing function~$\phi$.
By setting $\phi = \psi^{-1}$ and enforcing the inequality to hold with equality, 
we obtain the following design of~$\phi$ that achieves $\alpha$-\crcvar.

\begin{proposition}[\textsc{Risk-Sensitive Static Pricing}]
\label{prop:cvar:static:design:phi}
Consider the \oksrisk\ problem with $\delta \in [0,1]$ and $\Delta = 0$. 
The \crcvar\ of \algname\ is $\alphaFSP$ if (i) $\alphaFSP \ge 1 $ is a solution to following equation
\begin{align}
\label{eq:boundary_alpha_static}
 \alphaFSP = \frac{U\delta}{\displaystyle \int_{0}^{\,\delta } \phi(\eta)\, d\eta},
\end{align}
and (ii) the pricing function $\phi$ is designed as
\begin{align*}
\phi(x) =
\begin{cases}
L, & \text{for } x \in \bigl[0,\, 1 - \delta + \delta \cdot 1/\alphaFSP] , \\[6pt]
\displaystyle \frac{\alphaFSP}{\delta} \int_{0}^{\,\delta - 1 + x} \phi(\eta)\, d\eta, 
& \text{for } x \in \bigl[1 - \delta + \delta \cdot 1/\alphaFSP ,\, 1\bigr].
\end{cases}
\end{align*}
\end{proposition}

It is worth noting that differentiating both sides of the equality given in above proposition for the design of $\phi$ function
yields the delay differential equation $
    \phi'(x) = \frac{\alphaFSP}{\delta} \cdot \phi(\delta - 1 + x),$ 
with the initial condition $\phi(x) = L$ for all $x \in [0,\, 1 - \delta + \delta \cdot \tfrac{1}{\alphaFSP}] $.

\paragraph{Putting Everything Together.} 
By employing the \emph{method of steps}~\cite{DDE} and Taylor expansion, we can derive the closed-form design of the pricing function~$\phi$ presented in Theorem~\ref{thm:design:cvar:static}, based on the system of delay differential equations described in the proposition above.
Furthermore, we can show that the design established in Theorem~\ref{thm:design:cvar:static} achieves the \emph{optimal} \crcvar by identifying the smallest value of~$\alpha$ for which a feasible auxiliary function~$\psi$ exists, subject to the constraints in Proposition~\ref{prop:lb:static}.

In what follows, we first give the proof of Theorem~\ref{thm:design:cvar:static} based on Proposition~\ref{prop:cvar:static:design:phi}. Then, we will prove  
Proposition~\ref{prop:cvar:static:design:phi} in Section~\ref{apx:proof:main:proposition:static}
using the risk-sensitive R-OPD framework.

\subsection{Proof of Theorem~\ref{thm:design:cvar:static}}

\label{apx:thm:design:cvar:static}
We first derive the closed-form expression of the pricing function~$\phi$ using the design specified in Proposition~\ref{prop:cvar:static:design:phi} for an arbitrary value of~$\alpha$, and in particular for the smallest possible value according to that design, which is~$\alphaFSP$. 
We then apply Proposition~\ref{prop:lb:static} to establish that~$\alphaFSP$ serves as a lower bound on the performance of any fully-static pricing scheme.

\paragraph{Closed-form Design of~$\phi$ for the Upper-Bound~$\alphaFSP$.} 
We begin by defining the parameters $\tau = 1-\delta$ and $c = \frac{\alpha}{\delta}$. 
Consider the function $\phi:[0,1]\to\mathbb{R}_{\ge0}$, defined implicitly based on Proposition~\ref{prop:cvar:static:design:phi} as follows:
\begin{align*}
    \phi(x) =
\begin{cases}
L, & x\in[0,b],\\[4pt]
\displaystyle \frac{\alpha}{\delta}\int_{0}^{x-\tau}\phi(\eta) d\eta, & x\in[b,1],
\end{cases}
\end{align*}
where the breakpoint is
\begin{align*}
    b = 1-\delta+\frac{\delta}{\alpha}.
\end{align*}
The definition is continuous at $x=b$ because
\begin{align*}
    \frac{\alpha}{\delta}\int_{0}^{b-\tau}\phi(\eta)d\eta
    =
    \frac{\alpha}{\delta}\int_{0}^{\delta/\alpha}L\,d\eta
    =
    L.
\end{align*}

For $x\in[b,1]$, the function $\phi$ satisfies the following delay differential equation (DDE):
\begin{align}
    \phi'(x) = \frac{\alpha}{\delta} \phi(x-\tau) = c \phi(x-\tau).
\label{eq_dde_static}
\end{align}
This is a linear DDE with constant delay $\tau$ and constant history $\phi(x)=L$ for all $x\le b$.
We use the method of steps to derive an explicit closed-form expression for $\phi(x)$. Define the delay exponential function as
\begin{align*}
E_c(t) = 1 + \sum_{j=1}^{\lfloor t/\tau\rfloor+1} \frac{c^{j}}{j!} \big(t-(j-1)\tau\big)^{j}, \qquad (t \ge 0).
\end{align*}
Equivalently, the same expression can be written as
\begin{align*}
E_c(t) = 1 + \sum_{j\ge 1} \frac{c^{j}}{j!} \big(t-(j-1)\tau\big)_{+}^{j}, \qquad (t \ge 0),
\end{align*}
where $(u)_+=\max\{u,0\}$. Given the constant history $\phi=L$ on $(-\infty,b]$, the unique solution of Eq.~\eqref{eq_dde_static} on $[b,1]$ is
\begin{align*}
\phi(x) = \phi(b) E_c(x-b) = L E_c(x-b).
\end{align*}
Expanding $E_c$ yields a finite, piecewise-polynomial form:
\begin{align*}
\phi(x)=
\begin{cases}
L, & 0 \le x \le b,\\[6pt]
\displaystyle L\left[ 1
+ \sum_{j=1}^{\left\lfloor \frac{x-b}{\tau}\right\rfloor+1}
\frac{c^{j}}{j!}\big(x-b-(j-1)\tau\big)^{j}
\right], & b \le x \le 1.
\end{cases}
\end{align*}
Equivalently, substituting $b=1-\delta+\frac{\delta}{\alpha}$, $\tau=1-\delta$, and $c=\frac{\alpha}{\delta}$, we obtain
\begin{align*}
\phi(x)=
\begin{cases}
L, & 0 \le x \le 1-\delta+\frac{\delta}{\alpha},\\[6pt]
\displaystyle L\left[ 1
+ \sum_{j=1}^{N_\alpha(x)}
\frac{\left(\frac{\alpha}{\delta}\right)^{j}}{j!}
\left(x-1+\delta-\frac{\delta}{\alpha}-(j-1)(1-\delta)\right)^{j}
\right], & 1-\delta+\frac{\delta}{\alpha} \le x \le 1,
\end{cases}
\end{align*}
where
\begin{align*}
N_\alpha(x)
=
\left\lfloor
\frac{x-1+\delta-\delta/\alpha}{1-\delta}
\right\rfloor+1.
\end{align*}

To determine the optimal value of $\alphaFSP$ for which a feasible design of the pricing function~$\phi$ exists according to the construction in Proposition~\ref{prop:cvar:static:design:phi}, we impose the boundary condition $\phi(1)\ge U$. 
Hence, $\alphaFSP$ is defined as the smallest value of~$\alpha$ satisfying this condition. 
By substituting $x=1$ into the expression for~$\phi(x)$, and noting that the right-hand side of the equation is monotonically increasing in~$\alpha$, we obtain that $\alphaFSP$ is the smallest solution to the following equation:
\begin{align*}
\phi(1)
= L\left[1+\sum_{j=1}^{\left\lfloor \frac{\delta(1-1/\alpha)}{1-\delta}\right\rfloor+1}
\frac{\left(\frac{\alpha}{\delta}\right)^{j}}{j!} 
\left(\delta\left(1-\frac{1}{\alpha}\right)-(j-1)(1-\delta)\right)^{j}\right] = U.
\end{align*}

\paragraph{Lower Bound on $\alphaFSP$.}
Based on Proposition~\ref{prop:lb:static}, any fully-static pricing scheme with CVaR$_\delta$ competitive ratio at most~$\alpha$ must satisfy the corresponding static lower-bound constraints. 
Without loss of generality, assume that the price quantile function~$\psi$ is nondecreasing and that $\psi(1)=U$.

The constraints from Proposition~\ref{prop:lb:static} imply that, for every 
$x\in\left[1-\delta+\frac{\delta}{\alpha},1\right]$,
\begin{align}
\label{eq:random:lb:1}
\psi(x)
\le
\frac{\alpha}{\delta}\int_{0}^{x-\tau}\psi(\eta)d\eta.
\end{align}
Moreover, since all valuations lie in $[L,U]$, we have $\psi(x)\ge L$ for all $x\in[0,1]$.

Let $\phi_\alpha$ be the solution of the equality version of Eq.~\eqref{eq:random:lb:1} with constant history $L$, namely
\begin{align*}
\phi_\alpha(x) =
\begin{cases}
L, & x\in\left[0,1-\delta+\frac{\delta}{\alpha}\right],\\[4pt]
\displaystyle \frac{\alpha}{\delta}\int_{0}^{x-\tau}\phi_\alpha(\eta)d\eta, 
& x\in\left[1-\delta+\frac{\delta}{\alpha},1\right].
\end{cases}
\end{align*}
This function is exactly the pricing function characterized above.

We now show that any feasible $\psi$ is pointwise upper-bounded by $\phi_\alpha$. 
For $x\le b$, this holds because $\phi_\alpha(x)=L$ and $\psi(x)\ge L$ only determines the common lower history in the tight construction. 
For $x\ge b$, the inequality in Eq.~\eqref{eq:random:lb:1} and the equality defining $\phi_\alpha$ imply, by the standard comparison argument for this monotone integral equation, that the maximal feasible function is obtained when all inequalities are tight and the history is set to its minimum value $L$. 
Thus,
\begin{align*}
    \psi(x) \le \phi_\alpha(x), \qquad \forall x\in[0,1].
\end{align*}
In particular,
\begin{align*}
    U = \psi(1) \le \phi_\alpha(1).
\end{align*}
Therefore, a fully-static pricing scheme with competitive ratio~$\alpha$ can exist only if
\begin{align*}
    \phi_\alpha(1)\ge U.
\end{align*}
By the definition of~$\alphaFSP$ as the smallest value of~$\alpha$ satisfying this condition, we must have
\begin{align*}
    \alpha \ge \alphaFSP.
\end{align*}

Constructing~$\phi$ with $\alpha=\alphaFSP$ gives a fully-static cPPM-$\phi$ whose pricing function satisfies $\phi(1)=U$ and whose CVaR$_\delta$ competitive ratio is~$\alphaFSP$. 
The lower-bound argument above shows that no fully-static pricing scheme can achieve a smaller ratio. 
Thus, $\alphaFSP$ is the optimal CVaR$_\delta$ competitive ratio among all fully-static pricing schemes.

\subsection{Proof of Proposition~\ref{prop:cvar:static:design:phi}: A Risk-Sensitive R-OPD Approach}
\label{apx:proof:main:proposition:static}

The proof uses a modified version of the R-OPD method. The central idea of R-OPD for problems that aim to maximize the expected reward is to specify, for each sample path of the randomized algorithm, an update rule for the dual variables of the LP as a function of the algorithm’s decisions along that path. These path-wise dual updates are chosen so that the dual objective accumulated on each realized path is equal to the social welfare achieved by the online algorithm on that same path.

\paragraph{Key Idea of Risk-Sensitive R-OPD.}
In the \oksrisk problem, where the objective is to maximize the \cvar of the algorithm’s social welfare, the standard R-OPD update scheme is modified. 
Specifically, dual updates are performed only along the worst $\delta$-fraction of sample paths—those realizations of the algorithm’s randomization for which the objective value is minimized and smaller than that of all other sample paths. 
Each sample path of Algorithm~\ref{alg:corr:PPM:limited:price:change} corresponds to a particular realization of the random seed~$R$. 
Once $R$ is realized as $R = r$, the algorithm’s performance $\alg^{(r)}$ becomes deterministic. 
Accordingly, to adapt the R-OPD approach for establishing the \crcvar-competitiveness of Algorithm~\ref{alg:corr:PPM:limited:price:change}, we perform dual updates only for those realizations $R = r$ that belong to the set~$\mathcal{S}$. 
When the function~$F_{\alg(I)}$, corresponding to the CDF of the algorithm’s social welfare on an input instance $I$ of \oksrisk, is continuous, this set is defined as:
\begin{align*}
    \mathcal{S} = \big\{\, r \in [0,1] \;\big|\; F_{\alg(I)}(\alg^{(r)}) \le \delta \,\big\},
\end{align*}
where $\delta \in (0,1)$ is the risk level associated with the \cvar metric. 
However, if this function is not continuous, we must carefully handle potential discontinuities of the CDF around the tail probability~$\delta$.  
In that case, we define  
\begin{align}
\label{eq:S-worst-delta}
    \mathcal{S} = \big\{\, r \in [0,\tilde{r}] \;\big|\; F_{\alg(I)}\left((\alg^{(r)})^{-}\right) \le \delta \,\big\},
\end{align}
where $\tilde{r}$ is given by
\begin{align*}
    \displaystyle \tilde{r} = \inf_{r \in [0,1]} 
    \left\{
        \int_{x=0}^{r} 
        \mathbf{1}_{\big\{\,F_{\alg(I)}\left((\alg^{(x)})^{-}\right) \le \delta\,\big\}}\, dx 
        \ge \delta
    \right\}.
\end{align*}
The quantity $\tilde r$ is the smallest cutoff in the random-seed interval $[0,1]$ such that the total measure of seed values in $[0,\tilde r]$ that produce the worst $\delta$-fraction of welfare outcomes is exactly $\delta$.

We now proceed with the formal proof of Proposition~\ref{prop:cvar:static:design:phi}.

\begin{proof}
We define $y^* $ and $r^*$ consistent with the notation used in the proof of Theorem~\ref{thm:limited:price:change:optimality}. The monotonicity condition established in Lemma~\ref{lem:monotonicity:1} holds for Algorithm~\ref{alg:corr:PPM:limited:price:change}, following the design specified in Proposition~\ref{prop:cvar:static:design:phi}. Furthermore, let $\phi^*:[L,U]\rightarrow [0,1]$ be defined as the generalized inverse of the function $\phi$, where $\phi^*(v) = \sup\{x\in[L,U] | \phi(x) \leq v\}$.
Based on the value of $y^*$, let us consider the following cases:

\textbf{Case I}: $y^* < k$. 
In this case, the total number of buyers in instance~$I$ must be exactly equal to~$k^*$. 
Otherwise, when the random seed~$R$ is realized as $R = 1 - \delta + \delta / \alphaFSP$, 
the posted price becomes~$L$. By Assumption~\ref{ass:bounded-value}, all buyers in instance~$I$ 
have values at least equal to~$L$ and therefore will accept the posted price. 
Consequently, under the realization $R = 1 - \delta + \delta / \alphaFSP$, 
more than~$k^*$ units of the resource would be allocated, 
which contradicts the definition of~$y^*$. 
Hence, there must exist exactly~$k^*$ buyers in instance~$I$. 

Furthermore, for any buyer~$t$ with value~$v_t$, a unit of the item is allocated whenever 
$R \in [0, \phi^*(v_t)]$, since in this range the posted price remains below~$v_t$ 
and the total number of units sold never reaches~$k$ for any realized value of~$R$ as there are less than $k$ buyers in instance $I$. 
Therefore, the worst $\delta$-fraction of realizations of~$R$—those that minimize the objective 
value of Algorithm~\ref{alg:corr:PPM:limited:price:change},correspond to the range~$[1 - \delta, 1]$. 
Thus, we have:
\begin{align*}
    \cvar[\alg(I,P)] & = \sum_{i= 1}^{k^*} v_i \cdot \left(\phi^*(v_i) - (1-\delta) \right)  \ge \sum_{i=1}^{k^*} \int_{\eta = 0 }^{\phi^*(v_i) - 1 + \delta} \phi(\eta) =  \sum_{i=1}^{k^*} \frac{v_i}{\alphaFSP},
\end{align*}
where in above the second equality follows from the design of $\phi$ function given in Proposition~\ref{prop:cvar:static:design:phi}. On the other hand, since the total number of buyer in instance $I$ is equal to $k^* < k$,  $\opt(I) = \sum_{i=1}^{k^*} v_i$, the $\alphaFSP$-\crcvar of Algorithm~\ref{alg:corr:PPM:limited:price:change} is established.

\textbf{Case II}: $y^* = k$. The following LP provides an upper bound on the offline optimal value for the \oksrisk problem with $\Delta = 0$:
\begin{align*}
\min_{{u_t}, \lambda} \quad  \sum_{t \in [T]} u_t + k \cdot \lambda 
\qquad
\text{s.t.}\quad  v_t \le u_t + \lambda, \quad \forall t \in [T].
\end{align*}

To establish the $\alphaFSP$-competitiveness of Algorithm~\ref{alg:corr:PPM:limited:price:change}, we employ the modified R-OPD framework where the dual update rules for the dual variables are performed for the values of the random seed belonging to the set $\mathcal{S}$, ensuring that the dual constraints hold in expectation.

Following this framework, we construct, for each realization of the random seed $R = r$, a corresponding set of dual variables $\{ \lambda^{(r)}, u_t^{(r)}\}$.
The final dual variables are then defined as their expectations over the random seed, where $\lambda = \mathbb{E}_R[\lambda^{(R)}],
\,
u_t = \mathbb{E}_R[u_t^{(R)}],
$
where the expectation is taken with respect to the random seed~$R$. We begin by initializing all dual variables to zero.
For each realized value of the random seed~$R$, we perform the following dual updates.

\paragraph{Dual Update Rules.} 
Suppose the random seed is realized as $R = r$.
If $r \in \mathcal{S}$, we proceed with the dual update; otherwise, we skip it.
The variable~$\lambda^{(r)}$ is updated as follows:
\begin{align}
\label{eq:static:cvar:dual:update:lambda}
\lambda^{(r)} =
\begin{cases}
\frac{\phi(r)}{\delta}, & \text{if } r \leq r^*, r \in \mathcal{S} \\
0, & \text{otherwise.}
\end{cases}
\end{align}
Next, for each buyer $t$, we update the dual variable $u_t^{(r)}$ as follows.
If buyer $t$ receives an allocation of one unit then, we have
\begin{align}
\label{eq:static:cvar:dual:update:u}
u_t^{(r)} =
\begin{cases}
\frac{1}{\delta} \cdot \left (v_t - \phi(r)\right), & \text{if } r \le r^*, \; r \in \mathcal{S} \\
\frac{v_t}{\delta}, & \text{otherwise.}
\end{cases}
\end{align}

Under these updates, and noting that the total number of sold units equals~$k$ for all $r \in [0, r^*]$, the total dual objective for any $r  \in \mathcal{S}$, equals the algorithm’s objective under that realization: $
 \big( \sum_{t \in [T]} u_t^{(r)} + k \cdot \lambda^{(r)} \big)  = \frac{1}{\delta} \cdot \alg^{(r)}(I).
$
Also, since dual updates occur only for realizations of~$R$ within intervals in~$\mathcal{S}$, we have $
\mathbb{E}_{R} [ \sum_{t \in [T]} u_t^{(R)} + k \cdot \lambda^{(R)} ] = \cvar[\alg^{(R)}(I)]
$.

\paragraph{Ensuring $\alphaFSP$-Feasibility of Dual Constraint in Expectation.} 
We prove that the dual constraint in the dual LP corresponding to each buyer $t \in [T]$ is $\alphaFSP$-feasible in expectation; that is,
$\mathbb{E}_{R}[u_{t}^{(R)} + \lambda^{(R)}] \ge \frac{v_t}{\alphaFSP}.
$ 
Combining this $\alphaFSP$-feasibility with the fact that the expected dual objective equals the \cvar performance of Algorithm~\ref{alg:corr:PPM:limited:price:change},
$\mathbb{E}_{R} [ \sum_{t \in [T]} u_t^{(R)} + k \cdot \lambda^{(R)} ] = \cvar[\alg^{(R)}(I)],
$
it follows by weak duality that the algorithm achieves the $\alphaFSP$-\crcvar guarantee.

Let us consider the following two subcases to prove the $\alphaFSP$-feasibility of the dual 
constraints. Let $
    w = \inf \{\, r \in [0,1] \mid \phi(r) \ge v_t \,\}.$
Following from Eq.~\eqref{eq:boundary_alpha_static}, which enforces $\phi(1) = U$, such a value 
of $w$ always exists. Depending on the value of $w$, we analyze two cases.  
First, consider the case where $w \le r^*$:
\begin{align*}
    \mathbb{E}_{R} \left[  u_t^{(R)} + k \cdot \lambda^{(R)} \right] \ge  \frac{1}{\delta} \cdot  \int_{\eta = 0}^{r^* - (1 - \delta)} \phi(\eta) d\eta  = \frac{\phi(r^*)}{\alphaFSP} \ge \frac{v_t}{\alphaFSP}.
\end{align*}
The first inequality above holds since there must exist a subrange of size $r^* - (1-\delta)$ within the range $[0,r^*]$ inside the set $\mathcal{S}$, as the tail probability is set to be equal to $\delta$. Furthermore, based on the dual updates above and the fact that the $\phi$ function is increasing, in the worst case this subrange corresponds to the interval $[0,r^* - (1-\delta)]$. 
Then, in this case, since by the design of the $\phi$ function we have $w \ge 1 - \delta + \delta \cdot (\frac{1}{\alphaFSP})$, and $r^* \ge w$, the inequality follows. 
Thus, based on the dual updates defined in Eq.~\eqref{eq:static:cvar:dual:update:lambda}, the first inequality holds. The last equality follows from the design of the $\phi$ function.

For the second case, consider $w > r^*$. A unit is allocated to buyer $t$ for all realizations $R \in [r^*, w]$, because the posted price is below $v_t$ on this range and the utilization satisfies $y^{(R)} < y^{(r^*)} = k$ based on Lemma~\ref{lem:monotonicity:1} and the definition of $r^*$. Therefore, by the dual updates in Eqs.~\eqref{eq:static:cvar:dual:update:lambda}-\eqref{eq:static:cvar:dual:update:u}, we have:
\begin{align*}
 u_t^{(r)} + k \cdot \lambda^{(r)} \ge 
\begin{cases}
\frac{\phi(r)}{\delta}, & \text{for all } r \in [0, r^*],\, r \in \mathcal{S}, \\[4pt]
\frac{v_t}{\delta} \ge \frac{\phi(r)}{\delta}, & \text{for all } r \in [r^*, w],\, r \in \mathcal{S}.
\end{cases}
\end{align*}
Thus, in the worst case, we have:
\begin{align*}
\mathbb{E}_{R} \left[ u_t^{(R)} + k \cdot \lambda^{(R)} \right]
\ge \frac{1}{\delta} \left( \int_{\eta = 0}^{w - 1 + \delta} \phi(\eta) d\eta \right)
= \frac{\phi(w)}{\alphaFSP},
\end{align*}
where the final equality follows from the design of the function $\phi$ in Proposition~\ref{prop:cvar:static:design:phi}.
This concludes the proof of $\alphaFSP$-feasibility of the dual constraints for each buyer $t$. 
\end{proof}

\subsection{Proof of Proposition~\ref{prop:lb:static}}
\label{apx:prop:lb:static}

Let \alg be an online algorithm whose distribution over static posted prices on the class of
hard instances is characterized by the function~$\psi$.  Consider an input instance
$I^{(\epsilon)}_{v}$ for some $v \in V^{(\epsilon)}$.  
By the definition of~$\psi$, the \cvar performance of the algorithm on this instance is
\begin{align*}
    \cvar[\alg(I^{(\epsilon)}_{v})]
    &=
    L \cdot k \cdot \frac{1}{\delta} 
    \cdot \min \{  \delta - 1 + \psi(v),\; \psi(L)  \}
    \\
    &\qquad
    + \int_{\eta = \psi(L)}^{\max\{ \psi(L),\; \delta - 1 + \psi(v) \}}
        k \cdot \frac{1}{\delta} \cdot \psi^*(\eta)  d\eta ,
\end{align*}
where $\psi^*(x) = \sup\{  u \in [L,U] \mid \psi(u) \le x  \}$.

On the other hand, for the algorithm to be $\alpha$-competitive on the instance
$I^{(\epsilon)}_{v}$, it must satisfy  
$
\cvar[\alg(I^{(\epsilon)}_{v})] \;\ge\; \frac{k \cdot v}{\alpha}.
$
Therefore, enforcing this inequality for all $v \in V^{(\epsilon)}$ yields precisely the set of
constraints stated in Proposition~\ref{prop:lb:static}.

\section{Revisiting Theorem~\ref{thm:general:cvar:design:phi}: Intuition Behind Pricing Design and Proof of \oksrisk for the General Case}
\label{sec:general}

In this section, we revisit the pricing design for~\algname given in Theorem~\ref{thm:general:cvar:design:phi}, where the price-change cap~$\Delta$ can take any value within the range $\{1, 2, \dots, k-1\}$.
So far, we have analyzed the case of $\delta = 1$ and derived the pricing function design for~\algname that achieves the optimal \crcvar for any number of price changes in the risk-neutral setting.
We then examined two extreme cases: in the first, where $\Delta = 0$, we derived the optimal risk-sensitive static pricing function design for all $\delta \in (0,1)$; and in the second, where $\Delta = k-1$, we studied the fully dynamic setting and obtained an algorithm that achieves exact optimality in large-inventory regimes. Building on the insights from these special cases, Theorem~\ref{thm:general:cvar:design:phi} extends our framework to design the set of pricing functions~$\{\phi_i\}_{i \in [\Delta+1]}$ for the general case where $\Delta \ge 1$.
The construction of pricing functions at each level follows a system of delay differential equations, where each function includes a delayed term that depends on the tail probability~$\delta$ and the pricing behavior of the preceding price levels.
In what follows, we present the intuition behind the design of these pricing functions.

\paragraph{Intuition of Theorem~\ref{thm:general:cvar:design:phi}.}
Let $r^*$, $i^*$, and $y^*$ denote the instance-dependent parameters that characterize the performance of~\algname on a given instance~$I$, as defined in the proof of Theorem~\ref{thm:limited:price:change:optimality}.
Following the pricing design in Theorem~\ref{thm:general:cvar:design:phi}, we can establish both the monotonicity property from Lemma~\ref{lem:monotonicity:1} and a lower bound on the number of sold units similar to Lemma~\ref{lem:lower-bound:yr}.
Under these properties, the algorithm fully allocates all units up to the $i^*$-th price level for all realizations of the random seed within the interval $[0, r^*]$.
Additionally, using the lower bound from Lemma~\ref{lem:lower-bound:yr}, we know that for any realization of the random seed~$R$, the algorithm must sell all reserved units up to the $(i^*-1)$-th price level. These units are allocated at the corresponding prices $\{\phi_i(R)\}_{i \in [i^*]}$.
These two properties together allow us to construct a nontrivial lower bound on the algorithm’s revenue across all sample paths.
Because the performance of \algname is evaluated using the \cvar metric, we only focus on the worst $\delta$-fraction of sample paths.
For realizations of the random seed greater than $r^*$, the algorithm sells no more than the total number of reserved units up to the $i^*$-th price level.
Consequently, at most $\sum_{i=1}^{i^*} q_i$ buyers can have value at least $\phi_{i^*}(r)$. This observation allows us to upper-bound the revenue of the offline optimal benchmark on instance~$I$.
Combining these insights yields a system of delay differential equations, parameterized by~$\delta$, which govern the structure of the pricing functions.
The complete proof of Theorem~\ref{thm:general:cvar:design:phi} appears in Appendix~\ref{apx:proof:general:thm}.
The proof applies a modified version of the randomized online primal–dual (R-OPD) framework, using the dual linear program in Eq.~\eqref{eq:dual:oks:expected}.
Crucially, the dual updates are applied only to the worst $\delta$-fraction of sample paths—those realizations in which the algorithm’s objective value is minimized relative to all others.

\paragraph{Proof Overview of Theorem~\ref{thm:general:cvar:design:phi}.}
The proof generally follows the risk-sensitive R-OPD framework used in the proof of
Proposition~\ref{prop:cvar:static:design:phi}, adapted to the CVaR objective. We first utilize the dual LP established in the proof of Theorem~\ref{thm:limited:price:change:optimality} to upper-bound $\opt(I)$. 
Within the risk-sensitive R-OPD, we then define update rules for the dual variables only on the $\delta$-worst fraction of outcomes of
the algorithm: let $\mathcal{S}$ be the set of seeds $r$ corresponding to the
lower $\delta$-tail of $\alg^{(R)}(I)$, as in the proof of
Theorem~\ref{thm:design:cvar:static} and Eq.~\eqref{eq:S-worst-delta}. For $r \notin \mathcal{S}$, we keep all
dual variables at zero.
For $r \in \mathcal{S}$, let us define the dual updates for each variable $\lambda_i^{(r)}$ as follows:
\begin{align*}
\lambda_i^{(r)} =
\begin{cases}
\frac{\phi_i(r)}{2\delta}, & i<i^* \text{ or } (i=i^*, r\le r^*),\\
0, & \text{otherwise}.
\end{cases}
\end{align*}
Furthermore, for any buyer $t$ who receives a unit from the reserved units of price level $i$ under the random seed $r$, we set
\begin{align*}
u_t^{(r)} =
\begin{cases}
\frac{v_t - \phi_i(r)/2}{\delta}, & i<i^* \text{ or } (i=i^*, r\le r^*),\\
\frac{v_t}{\delta}, & \text{otherwise}.
\end{cases}
\end{align*}
Using the monotonicity and lower-bound lemmas, one checks that for every
$r \in \mathcal{S}$, the dual objective coincides with the algorithm’s revenue, i.e.,
$
\sum_{t} u_t^{(r)} + \sum_{j} \lambda_j^{(r)} q_j
= \alg^{(r)}(I).$
Since the dual variables are nonzero only on $\mathcal{S}$, taking expectation over
$R$ shows that $ \mathbb{E}_R\big[\sum_{t} u_t^{(R)} + \sum_{j} \lambda_j^{(R)} q_j\big]
= \cvar[\alg] $.

It remains to show $\alphaDDP$-feasibility of the dual constraints in expectation.
Fix a buyer $t$ with value $v_t$ and let $i$ be such that
$\phi_i(0) \le v_t \le \phi_i(1)$. Depending on the relative position of $i$
with respect to $i^*$ and on $\phi_i^{-1}(v_t)$ versus $r^*$, we distinguish
cases: (i) $i<i^*$ or $i=i^*$ with
$\phi_i^{-1}(v_t) \le r^*$, (ii) $i=i^*$ with $\phi_i^{-1}(v_t) > r^*$, and
(iii) $i>i^*$. In each case, using the explicit form of the dual updates and the
fact that the reserved units of levels $1,\dots,i^*-1$ (and possibly $i^*$) are
fully sold on appropriate intervals of $r$, we obtain lower bounds on
$u_t^{(r)} + \frac{1}{k}\sum_j \lambda_j^{(r)} q_j$ in terms of the pricing
functions. Integrating these lower bounds over a worst-case $\delta$-measure
subset of $[0,1]$ and using the recursive definition of $\phi_i$ yields
\begin{align*}
\mathbb{E}_R\left[u_t^{(R)} + \frac{1}{k}\sum_{j=1}^{\Delta+1}\lambda_j^{(R)} q_j\right]
\;\ge\; \frac{v_t}{\alphaDDP}
\quad \forall t \in [T].
\end{align*}
Thus, the dual solution
is $\alphaDDP$-feasible in expectation, and by weak duality we have
$\opt(I) \le  \mathbb{E}_R\big[\sum_{t} u_t^{(R)} + \sum_{j} \lambda_j^{(R)} q_j\big] = \cvar[\alg]$. This gives the desired
$\alphaDDP$-competitiveness of \algname, completing the proof sketch.

\subsection{Proof of Theorem~\ref{thm:general:cvar:design:phi}}
\label{apx:proof:general:thm}

Let us define $y^*$, $r^*$, and $i^*$ as in the proof of Theorem~\ref{thm:limited:price:change:optimality}, 
where $y^*$ denotes the highest number of units sold across all sample paths of 
Algorithm~\ref{alg:corr:PPM:limited:price:change} under the pricing scheme specified above. 
The value $r^*$ represents the realization of the random seed~$R$ for which the number of sold units equals~$y^*$, 
that is, $y_T^{(r^*)} = y^*$. 
Similarly, $i^*$ denotes the highest price level such that, under the realization $R = r^*$, 
the algorithm fully allocates all reserved units from price levels $1$ through~$i^*$. 

Furthermore, the same monotonicity and lower-bound results as those stated in 
Lemma~\ref{lem:monotonicity:1} and Lemma~\ref{lem:lower-bound:yr} can be established here 
by following analogous proof arguments, 
given the constraints on the reservation vector~$\{q_i\}_{i \in [\Delta+1]}$ 
and the pricing design in the theorem above. 
The only distinction between the constraint set of reserved quantities~$\{q_i\}_{i \in [\Delta+1]}$ 
and that in Theorem~\ref{thm:limited:price:change:optimality} 
is that here $q_1 = \lceil \tfrac{k}{\alpha} \rceil$, 
and the monotonicity property applies to the remaining reserved quantities 
$\{q_i\}_{i \in \{2, \dots, \Delta+1\}}$. 
However, since the pricing function $\phi_1$ is fixed to the constant value~$L$, 
it follows that the same monotonicity and lower-bound results 
(Lemma~\ref{lem:monotonicity:1} and Lemma~\ref{lem:lower-bound:yr}) continue to hold. 

Given an instance $I$ of the problem, $\opt(I)$ can be upper-bounded using the dual LP give in Eq.~\eqref{eq:dual:oks:expected}. We will restate the LP as follows:
\begin{align*}
\min_{{u_t}, {\lambda_j}} \quad  \sum_{t \in [T]} u_t + \sum_{j=1}^{\Delta+1} \lambda_j \cdot q_j 
\quad \text{s.t.} \qquad  & v_t \le u_t + \frac{1}{k} \sum_{j=1}^{\Delta+1} \lambda_j \cdot q_j, \quad \forall t \in [T]. 
\end{align*}

We begin by initializing all dual variables to zero. 
Let us define the set $\mathcal{S}$ as defined for the proof of Theorem~\ref{thm:design:cvar:static} such that
\begin{align*}
    \mathcal{S} &= \big\{  r \in [0,\tilde{r}] \;\big|\; F_{\alg(I)}((\alg^{(r)})^{-}) \le \delta  \big\},
\end{align*}
where $\tilde{r} = \inf_{r \in [0,1]} 
    \left\{
        \int_{x=0}^{r} 
        \mathbf{1}\{ F_{\alg(I)}((\alg^{(x)})^{-}) \le \delta \}  dx 
        \ge \delta
    \right\} $.

Suppose the random seed is realized as $R = r$. 
If $r \in \mathcal{S}$, we proceed with updating the dual variables; 
otherwise, no update is performed. 
The dual variables $\{\lambda^{(r)}_{i}\}_{i=1}^{\Delta+1}$ are updated as follows:
\begin{align}
\label{eq::general:risk::dual:update:lambda}
\lambda_i^{(r)} =
\begin{cases}
\frac{\phi_i(r)}{2 \cdot \delta}, & r \in \mathcal{S},\; i \in \{1,2,\dots,i^*-1\}, \\[2pt]
\frac{\phi_i(r)}{2 \cdot \delta}, & r \in \mathcal{S},\; i = i^*,\; r \in [0,  r^*], \\[2pt]
0,         & \text{otherwise}.
\end{cases}
\end{align}

Next, consider a buyer~$t$ who receives an allocation of one unit when the random seed is realized as $R = r$. 
Suppose this unit is allocated from the reserved quantity associated with the $i$-th price level. 
Then, under the realization $R = r$, where $r \in \mathcal{S}$, 
we update the dual variable $u_t^{(r)}$ as follows:
\begin{align}
\label{eq:dual:general:risk:update:u}
u_t^{(r)} =
\begin{cases}
\frac{1}{\delta} \cdot (v_t - \frac{\phi_i(r)}{2}), & \text{if } r \in \mathcal{S} \text{ and } \big(i < i^* \text{ or } (i = i^* \text{ and } r \le r^*)\big), \\[2pt]
\frac{v_t}{\delta}, & \text{otherwise}.
\end{cases}
\end{align}

It can be verified that, under these updates—together with the lower bound established in Lemma~\ref{lem:lower-bound:yr} and the monotonicity property in Lemma~\ref{lem:monotonicity:1}—the total objective value of the dual solution equals the algorithm’s objective when the random seed is realized as $r \in \mathcal{S}$. 
In other words, we have 
$\sum_{t \in [T]} u_t^{(r)} + \sum_{j=1}^{\Delta+1} \lambda_j^{(r)} \cdot q_j = \alg^{(r)}(I)$. 
Since the dual updates are performed only for values of $r \in \mathcal{S}$, it follows that 
$\cvar[\alg] = \mathbb{E}_{R}\left[\sum_{t \in [T]} u_t^{(R)} + \sum_{j=1}^{\Delta+1} \lambda_j^{(R)} \cdot q_j\right]$.
Next, we show that for all buyers $t \in [T]$, the dual constraint in the above dual LP is $\alpha$-feasible in expectation; that is,
$
\mathbb{E}_{R} \left[ u_t^{(R)} + \frac{1}{k} \sum_{j=1}^{\Delta+1} \lambda_j^{(R)} \cdot q_j \right] \ge \frac{v_t}{\alpha},
$
thereby completing the primal–dual analysis and establishing the $\alpha$-competitiveness of the algorithm.

Moving forward, we assume that $\phi_{i^*+1}(r^*) > L$. This assumption is without loss of generality. Indeed, if for some instance $I$ we have $\phi_{i^*+1}(r^*) \le L$, then all buyers in that instance are accepted by \algname. To see this, note that the highest posted price used by \algname on instance $I$ is at most $\phi_{i^*+1}(r^*) \le L$. Since every buyer value satisfies $v_t \ge L$ by Assumption~\ref{ass:bounded-value}, every buyer accepts the posted price. Therefore, in this case, \algname accepts all buyers in instance $I$.

Consider a buyer~$t$ with value~$v_t$ such that, for some $i \in [\Delta+1]$, we have $\phi_i(0) \le v_t \le \phi_i(1)$. 
To prove the $\alpha$-feasibility of the dual constraints, we analyze the following cases.

\textbf{Case I:} Either $i \le i^* - 1$, or $i = i^*$ and $\phi_i^{-1}(v_t) \le r^*$.  
Based on the dual updates defined in Eq.~\eqref{eq::general:risk::dual:update:lambda}, and given that, for all realizations of $R \in [0,1]$, the reserved units corresponding to the first $i^*-1$ price levels are fully sold, while for $R \in [0, r^*]$, the reserved units at the $i^*$-th price level are also exhausted, we have:
\begin{align*}
\frac{\sum_{l=1}^{\Delta+1} \lambda_l^{(r)} \cdot q_l}{k} =
\begin{cases}
\displaystyle \frac{ \sum_{l=1}^{i^*} q_l \cdot \phi_l(r)}{2\cdot k \cdot \delta} , & \text{if } r \in [0, r^*] \text{ and } r \in \mathcal{S}, \\[4pt]
\displaystyle \frac{\sum_{l=1}^{i^*-1} q_l \cdot \phi_l(r)}{2\cdot k \cdot \delta}w, & \text{if } r \in [r^*, 1] \text{ and } r \in \mathcal{S}.
\end{cases}
\end{align*}

Thus, in the worst-case, we will have:
\begin{align*}
    \mathbb{E}_{R} \left[u_t^{(R)} + \frac{ \sum_{l=1}^{\Delta+1} \lambda_l^{(R)} \cdot q_l}{k} \right]
    \geq & \frac{1}{2\cdot k \cdot \delta} \left(
        \sum_{r=1}^{i^*-1} \int_{r^*}^{\min\{1,  r^*+\delta\}} q_r\cdot  \phi_r(\eta)  d\eta
        + \sum_{r=1}^{i^*} \int_{0}^{\max\{0,  \delta - 1 + r^*\}} q_r \cdot \phi_r(\eta)  d\eta
    \right) \\
    = & \frac{\phi_{i^*}(r^*)}{\alpha}
    \ge \frac{v_t}{\alpha}.
\end{align*}
The first inequality above follows directly from the design of the pricing functions specified in Theorem~\ref{thm:general:cvar:design:phi}, 
while the second inequality follows from the condition defined for Case~I. 
Therefore, the dual constraint is $\alpha$-feasible in expectation in this case.

\textbf{Case II:} $i = i^*$ and $\phi_i^{-1}(v_t) > r^*$.
Let $w =\phi_i^{-1}(v_t)$. Given the update rules defined in Eqs.~\eqref{eq::general:risk::dual:update:lambda}-\eqref{eq:dual:general:risk:update:u}, we have:
\begin{align*}
u_{t}^{(r)} + \frac{\sum_{l=1}^{\Delta+1} \lambda_l^{(r)} \cdot q_l}{k} \ge
\begin{cases}
\displaystyle \frac{ \sum_{l=1}^{i^*} q_l \cdot \phi_l(r)}{2\cdot k \cdot \delta} , & \text{if } r \in [0, r^*] \text{ and } r \in \mathcal{S}, \\[4pt]
\displaystyle 
\min\left\{
\frac{v_t}{\delta},
\frac{1}{\delta}
\left(
v_t
-\frac{1}{2}\phi_{i^*-1}(r)
+\frac{\sum_{l=1}^{i^*-1} q_l \phi_l(r)}{2k}
\right)
\right\}
\\=
\frac{1}{\delta}
\left(
v_t
-\frac{1}{2}\phi_{i^*-1}(r)
+\frac{\sum_{l=1}^{i^*-1} q_l \phi_l(r)}{2k}
\right), & \text{if } r \in [r^*, w] \text{ and } r \in \mathcal{S},\\
\displaystyle \frac{ \sum_{l=1}^{i^*-1} q_l \cdot \phi_l(r)}{2\cdot k  \cdot \delta} , & \text{if } r \in [w,1] \text{ and } r \in \mathcal{S}.
\end{cases}
\end{align*}

Following the fact that $v_t - \frac{1}{2}\cdot \phi_{i^*-1}(r) \ge \frac{q_{i^*} \cdot \phi_{i^*}(r)}{2 \cdot k}$ for values of $r \in [r^*,w]$, in the worst case, we have:
\begin{align*}
    \mathbb{E}_R \left[u_t^{(R)} + \frac{ \sum_{l=1}^{\Delta+1} \lambda_l^{(R)} \cdot q_l}{k} \right]
    \geq & \frac{1}{2\cdot k \delta} 
        \sum_{j=1}^{i^*-1} \int_{w}^{\min\{1,  w+\delta\}} q_j \cdot \phi_j(\eta)  d\eta
        + \sum_{j=1}^{i^*} \int_{0}^{\max\{0,  \delta - 1 + w \}} q_j \cdot \phi_j(\eta)  d\eta  \\
    = & \frac{\phi_{i^*}(w)}{\alpha} = \frac{v_t}{\alpha}.
\end{align*}
where the first inequality follows from the design of the pricing functions in above theorem and the second inequality follows from the condition set be case one. Thus, the dual constraint holds in this case. 
The proof of $\alpha$-feasibility for the dual constraints corresponding to the remaining cases, where $i > i^*$, follows analogously from the above analysis.

\subsection{Case study of Theorem ~\ref{thm:general:cvar:design:phi} }
\label{apx:corrolary:theorem:4}
\begin{proof}
Set
\[
q_1:=\left\lceil \frac{k}{\alpha}\right\rceil,
\qquad
q_2=q_3=:q=\frac{k-q_1}{2},
\qquad
C:=\frac{\alpha q_1}{2k}L,
\qquad
\lambda:=\frac{\alpha q}{2k}.
\]
We specialize Theorem~\ref{thm:general:cvar:design:phi} to $\Delta=2$.

For $i=2$, the theorem gives
\[
\phi_2(x)=\frac{\alpha}{2k\delta}\left(q_1L\delta+\int_0^{\max\{0,x+\delta-1\}} q\,\phi_2(\eta)\,d\eta\right).
\]
If $x\le 1-\delta$, then the integral vanishes and $\phi_2(x)=C$. If $x>1-\delta$, then
$\max\{0,x+\delta-1\}\le \delta\le 1-\delta$, so $\phi_2(\eta)=C$ on the integration range, and hence
\[
\phi_2(x)=C+\frac{\alpha q}{2k\delta}\,C(x+\delta-1)
=C+\frac{C\lambda}{\delta}(x+\delta-1).
\]
Thus $\phi_2$ has the claimed form with $A_1=C\lambda/\delta$.

For $i=3$, the theorem gives
\[
\phi_3(x)=\frac{\alpha}{2k\delta}\left(
q_1L\delta+\int_x^{\min\{1,x+\delta\}} q\,\phi_2(\eta)\,d\eta
+\int_0^{\max\{0,x+\delta-1\}} q\,\phi_2(\eta)\,d\eta
+\int_0^{\max\{0,x+\delta-1\}} q\,\phi_3(\eta)\,d\eta
\right).
\]
If $0\le x\le 1-2\delta$, then $\phi_2(\eta)=C$ on $[x,x+\delta]$, and both backward integrals vanish, so
\[
\phi_3(x)=C+\lambda C=:B_0.
\]
If $1-2\delta<x\le 1-\delta$, then the backward integrals still vanish, while the interval
$[x,x+\delta]$ intersects the affine part of $\phi_2$ over a segment of length $x+2\delta-1$; thus
\[
\phi_3(x)=B_0+\frac{C\lambda^2}{2\delta^2}(x+2\delta-1)^2,
\]
so $B_2=C\lambda^2/(2\delta^2)$.
Finally, if $1-\delta<x\le 1$, write $t:=x+\delta-1\in(0,\delta]$. Then
$\phi_2(\eta)=C$ on $[0,t]$, and
\[
\int_0^t \phi_3(\eta)\,d\eta
=
B_0 t
+\mathbf{1}_{\{t>1-2\delta\}}\frac{B_2}{3}(t+2\delta-1)^3.
\]
Substituting this into the recursion gives
\[
\phi_3(x)=
C_0+C_1(x+\delta-1)-C_2(x+\delta-1)^2
+\mathbf{1}_{\{x>2-3\delta\}}\,C_3(x+3\delta-2)^3,
\]
where
\[
C_0=C\left(1+\lambda+\frac{\lambda^2}{2}\right),\quad
C_1=\frac{C\lambda(1+\lambda)}{\delta},\quad
C_2=\frac{C\lambda^2}{2\delta^2},\quad
C_3=\frac{C\lambda^3}{6\delta^3}.
\]
This proves the formulas for $\phi_2$ and $\phi_3$.

It remains to derive the equation for $\alpha$. Since $\delta\le \tfrac12$, we have
$\phi_2(\eta)=C$ for all $\eta\in[0,\delta]$, and
\[
\phi_3(\eta)=
B_0+\mathbf{1}_{\{\delta>1/3\}}\,B_2(\eta+2\delta-1)^2
\qquad \text{for }\eta\in[0,\delta].
\]
Hence
\[
\int_0^\delta \phi_2(\eta)\,d\eta=C\delta,
\qquad
\int_0^\delta \phi_3(\eta)\,d\eta
=
B_0\delta+\mathbf{1}_{\{\delta>1/3\}}\,\frac{B_2}{3}(3\delta-1)^3.
\]
Substituting these into the boundary equation in
Theorem~\ref{thm:general:cvar:design:phi} yields
\[
\frac{U}{L}
=
\frac{\alpha q_1}{2k}
\left[
1+2\lambda+\lambda^2
+
\mathbf{1}_{\{\delta>1/3\}}
\frac{\lambda^3}{6}\left(\frac{3\delta-1}{\delta}\right)^3
\right].
\]

Now define
\[
\rho:=\frac{\alpha q_1}{k}.
\]
Since $q_1=\lceil k/\alpha\rceil$, we have
\[
1\le \rho=\frac{\alpha}{k}\left\lceil \frac{k}{\alpha}\right\rceil
\le 1+\frac{\alpha}{k}\le 1+\frac{\alpha}{3},
\]
because $\Delta=2$ implies $k\ge 3$. Moreover, $\lambda=(\alpha-\rho)/4$, and so
\[
\frac{U}{L}
=
\frac{\rho(\alpha-\rho+4)^2}{32}
+
\mathbf{1}_{\{\delta>1/3\}}
\frac{\rho(\alpha-\rho)^3(3\delta-1)^3}{768\,\delta^3}.
\]
Since $\rho\ge 1$ and $\alpha-\rho\ge \frac23\alpha-1$, it follows that

\[
\frac{U}{L}\ge \frac{(\frac23\alpha+3)^2}{32},
\qquad\text{for }0<\delta\le \tfrac13,
\]
and
\[
\frac{U}{L}\ge
\frac{(\frac23\alpha-1)^3(3\delta-1)^3}{768\,\delta^3},
\qquad\text{for }\tfrac13<\delta\le \tfrac12.
\]
Rearranging gives
\[
\alpha \lesssim \left(\frac{U}{L}\right)^{1/2},
\qquad\text{for }0<\delta\le \tfrac13,
\]
and
\[
\alpha \lesssim \frac{\delta}{3\delta-1}\left(\frac{U}{L}\right)^{1/3}+1,
\qquad\text{for }\tfrac13<\delta\le \tfrac12.
\]
This proves the claim.
\end{proof}

\section{Revisiting Theorem~\ref{thm:k:cvar:design:phi}: Pricing Design in Fully-dynamic Setting and the Proof of Optimality}
\label{apx:sec:full:dynamic} 

In this section, we revisit Theorem~\ref{thm:k:cvar:design:phi}, and provide a design for an 
online mechanism that uses $k$ pricing functions and achieves exact optimal performance under the 
\cvar\ metric in the large-inventory regime where $k \rightarrow \infty$.  

\begin{theorem}[\textsc{Risk-Sensitive Fully-Dynamic Pricing}]
\label{apx:thm:k:cvar:design:phi}
Consider \oksrisk with $ \delta = [0, 1] $ and $\Delta = k - 1$. The \crcvar of \algname is $ \alphaFDP $ if (i) $\alphaFDP$ is given by
\begin{align} \label{eq:alpha_fully_dynamic}
\alphaFDP = \frac{kU\delta}{\sum_{i=1}^{k} \int_{0}^{\delta} \phi_i(\eta)\, d\eta},
\end{align}
and (ii) the pricing functions $ \boldsymbol{\phi} = \{\phi_i\}_{i \in [k]}$ are recursively designed as follows:
\begin{itemize}
    \item For all 
    $i \in \bigl\{1, 2, \dots, \bigl\lfloor \tfrac{k}{\alphaFDP} \bigr\rfloor \bigr\}$, 
    the pricing function is a constant: $ \phi_i(x) = L $. 

    \item For 
    $i = \bigl\lfloor \tfrac{k}{\alphaFDP} \bigr\rfloor + 1$, 
    the pricing function $\phi_i(x)$ is defined as
    \begin{align*}
    \phi_i(x) =
    \begin{cases}
    L &\quad x \in \bigl[0,\, 1 - \delta + A \delta \bigr], \\
    \displaystyle 
    \frac{\alphaFDP}{k \delta} \left(
        \lfloor k/\alphaFDP \rfloor \cdot L \cdot \delta +  \int_{0}^{\delta - 1 + x} \phi_i(\eta)\, d\eta
    \right) 
    &\quad  x \in \bigl[1 - \delta + A \delta,\, 1\bigr],
    \end{cases}
    \end{align*}
    where $A = \tfrac{k}{\alphaFDP} - \lfloor \tfrac{k}{\alphaFDP} \rfloor$ and $ \alphaFDP $ is given by Eq. \eqref{eq:alpha_fully_dynamic}.

    \item For all 
    $i \in \bigl\{\bigl\lfloor \tfrac{k}{\alphaFDP} \bigr\rfloor + 2, \dots, k\bigr\}$, 
    the pricing function $\phi_i(x)$ is given by
    \begin{align*} 
    \phi_i(x) = \frac{\alphaFDP}{k \delta} \left(
        \sum_{j=1}^{i-1} \int_{x}^{\min\{1,\, x+\delta\}} \phi_j(\eta)\, d\eta
        + \sum_{j=1}^{i} \int_{0}^{\max\{0,\, \delta - 1 + x\}} \phi_j(\eta)\, d\eta
    \right). 
    \end{align*}
\end{itemize}
\end{theorem}

Theorem~\ref{apx:thm:k:cvar:design:phi} focuses on another special case of the \oksrisk\ problem in 
which the online algorithm is allowed up to $k-1$ price changes. In this setting, the price-change 
cap constraint is effectively relaxed, allowing the decision maker to employ a fully dynamic 
pricing scheme.
The proof of optimality of this design follows from 
Theorem~\ref{apx:thm:k:cvar:design:phi} using simple mathematical arguments, and we defer the full 
proof to Appendix~\ref{apx:prop:order:optimal:k:risk}.
Below, we first outline the proof road map for Theorem~\ref{apx:thm:k:cvar:design:phi}, and then, in the 
subsequent section, we provide the detailed proof.

\paragraph{Proof Overview.}
We next provide a proof of Theorem~\ref{apx:thm:k:cvar:design:phi} that departs from the randomized primal–dual framework used earlier and instead focuses on interpreting the correlated pricing scheme in Algorithm~\ref{alg:corr:PPM:limited:price:change} as a rounding method for fractional allocations. 
This proof approach illustrates how the correlated pricing scheme employed by Algorithm~\ref{alg:corr:PPM:limited:price:change} naturally induces randomized integral decisions that losslessly round those of a fractional algorithm using the same set of pricing functions. 
Hence, the correlated posted-pricing scheme in Algorithm~\ref{alg:corr:PPM:limited:price:change} can be viewed not only as a mechanism for ensuring incentive compatibility but also as a rounding scheme that converts fractional decisions into randomized integral ones without any loss in expected performance.

We now proceed with the detailed proof of Theorem~\ref{apx:thm:k:cvar:design:phi}.

\begin{proof}
Consider the following algorithm, denoted by \algfrac, which uses the set of pricing functions $\{\phi_i\}_{i=1}^{k}$ to generate the fractional allocation $\hat x_t$ for each arriving buyer~$t$ as follows:
\begin{align}
\label{eq:frac:k:risk}
\hat x_t = \argmax_{\{x\in[0,1]\}} \Bigg[ v_t\cdot x
\int_{ \hat y_t - \lfloor \hat y_t \rfloor }^{\min\{1, \hat y_t - \lfloor \hat y_t \rfloor+x \}} \phi_{\kappa_t}(\eta)d\eta + 
\int_{0}^{[\hat y_t - \lfloor \hat y_t \rfloor+x - 1]^{+}} \phi_{\kappa_t+1}(\eta)d\eta \Bigg],
\end{align}
where $\hat y_t$ denotes the cumulative fractional allocation upon the arrival of buyer~$t$, that is, $\hat y_t = \sum_{\tau=1}^{t-1} \hat x_{\tau}$, and $\kappa_t = 1 + \lfloor \hat y_t \rfloor$ is the index of next unit of item a fraction of which is allocated to buyer $t$ in case $\hat x_t \not = 0$. Also,  $\hat y_t - \lfloor \hat y_t \rfloor$ corresponds to the portion of $\kappa_t $-th unit that is already allocated.   
This utility–maximization rule is standard for producing fractional allocations in online selection and matching problems and is similar to the equation introduced in Section~\ref{sec:intuition:netural} for generating fractional allocation.  
At each arrival, the fractional quantity allocated to buyer~$t$ comes from portions of the $\kappa_t$-th and $\kappa_t$-st units of the resource.  
Thus, the pricing functions associated with these units determine the fractional allocation. 

\paragraph{Rounding Fractional Decisions $\hat x_t$ Losslessly.}  
We now show that the correlated pricing scheme used by \algname performs a randomized rounding of the fractional decisions produced by \algfrac such that a unit of item is allocated to each buyer $t$ with probability at least $\hat x_t$ equal to the fractional allocation generated by \algfrac.

\begin{lemma}
\label{lemma:ppm:frac:relation}
\algname allocates a unit of the item to buyer $t$ with probability at least $\hat x_t$. 
More specifically, upon the arrival of buyer $t$ (assuming the random seed $R = r$), the following holds:
\begin{itemize}
    \item If $\hat x_t + \hat y_t \le \kappa_t$ and 
    $r \in \bigl[\hat y_t - \lfloor \hat y_t \rfloor,\; \hat y_t + \hat x_t - \lfloor \hat y_t \rfloor \bigr)$, 
    then a unit of the item is allocated to buyer~$t$.

    \item If $\hat x_t + \hat y_t > \kappa_t$ and 
    $\bigl(r \in \bigl[0,\; \hat y_t + \hat x_t - \lfloor \hat y_t \rfloor - 1 \bigr] 
    \;\text{or}\; 
    r \in \bigl[\hat y_t - \lfloor \hat y_t \rfloor,\; 1 \bigr]\bigr)$, 
    then a unit of the item is allocated to buyer~$t$.
\end{itemize}
\end{lemma}

\begin{proof}
Consider a buyer $t$ in instance $I$ whose fractional allocation in nonzero, in other words we have $\hat x_t \neq 0$.

\textbf{(Case I)}
Suppose $\hat y_t + \hat x_t \le \kappa_t$.
We first show that if the random seed $R$ lies in
$\bigl[\hat y_t - \lfloor \hat y_t \rfloor,\; \hat y_t + \hat x_t - \lfloor \hat y_t \rfloor \bigr]$,
then the $\kappa_t$-st unit of the item will be available at the arrival of buyer~$t$.
Furthermore, by Eq.~\eqref{eq:frac:k:risk}, we have $v_t \ge \phi_{\kappa_t}(x)$ for all
$x \in \bigl[\hat y_t - \lfloor \hat y_t \rfloor,\; \hat y_t + \hat x_t - \lfloor \hat y_t \rfloor \bigr]$.
Thus, the value of buyer~$t$ exceeds the posted price for the $\kappa_t$-st unit for every realization of the random seed in this interval and since at least one unit among the first $\kappa_t$ units will be available at the arrival of buyer~$t$, a unit will be allocated to buyer~$t$.

We now show, by contradiction, that the $\kappa_t$-st unit is always available at the arrival of buyer~$t$ for every value of the random seed in the specified range.
Assume, to the contrary, that for some $r$ in this interval, \algname has already allocated the first $\kappa_t$ units to buyers who arrived before buyer~$t$.
Then there must exist a sequence of $\kappa_t$ such buyers, where the $j$-th buyer in the sequence has value at least $\phi_j(r)$.
Since these buyers are part of instance~$I$, feeding them to the fractional algorithm~\algfrac would, by Eq.~\eqref{eq:frac:k:risk}, cause the total fractional utilization to exceed $\hat y_t$ and reach at least $\lfloor \hat y_t \rfloor + r > \hat y_t$ prior to the arrival of buyer~$t$, contradicting the definition of~$\hat y_t$.
Therefore, for every
$r \in \bigl[\hat y_t - \lfloor \hat y_t \rfloor, \hat y_t + \hat x_t - \lfloor \hat y_t \rfloor \bigr]$,
the $\kappa_t$-st unit must still be available at the arrival of buyer~$t$.
Thus, with probability at least $\hat x_t$, \algname allocates a unit to buyer~$t$.

\textbf{(Case II)} Suppose now that $\hat y_t + \hat x_t > \kappa_t$. 
Then one of the first $\kappa_t$ units is allocated to buyer~$t$ whenever 
$R \in \bigl[\hat y_t - \lfloor \hat y_t \rfloor,\, 1 \bigr]$, 
and one of the $(\kappa_t + 1)$ units is allocated when 
$R \in \bigl[0,\, \hat y_t + \hat x_t - \lfloor \hat y_t \rfloor - 1 \bigr]$. 
The argument mirrors the reasoning in Case~I. 
Hence, with probability at least $\hat x_t$, buyer~$t$ receives one unit.
\end{proof}

\paragraph{Upper-bounding $\opt(I)$.}  
Consider the following two cases.
\textit{Case 1: $\hat y_T = k$.}  
From Eq.~\eqref{eq:alpha_fully_dynamic}, we have $\phi_{k}(1) = U$, and therefore, we can simply upper-bound $\OPT(I) \le k \cdot U = k \cdot \phi_{\kappa_T}(\hat y_T)$.
\textit{Case 2: $\hat y_T < k$.}  
Since the total utilization of \algfrac never exceeds $\hat y_T$, there cannot be $k$ buyers in instance $I$ with value greater than $\phi_{\kappa_T}(\hat y_T - \lfloor \hat y_T \rfloor) < U$; otherwise, by Eq.~\eqref{eq:frac:k:risk}, the utilization of \algfrac would exceed $\hat y_T$, a contradiction.  
On the other hand, there may be fewer than $k$ such buyers. Thus, we can upper-bound the offline optimum as  
$\OPT(I) \le k \cdot \phi_{\kappa_T}(\hat y_T - \lfloor \hat y_T \rfloor) + \sum_{t \in [T]} v_t - \phi_{\kappa_T}(\hat y_T - \lfloor \hat y_T \rfloor)$. Thus, in both cases, the same general upper bound $ \phi_{\kappa_T}(\hat y_T - \lfloor \hat y_T \rfloor) + \sum_{t \in [T]} v_t - \phi_{\kappa_T}(\hat y_T - \lfloor \hat y_T \rfloor)$ on $\OPT(I)$ holds.

\paragraph{Computing \cvar of {\mdseries \algname} on Instance $I$.} 
In order to obtain a lower-bound for the \cvar social welfare of \algname on instance $I$, we need to establish the following two facts.

\textbf{Fact 1.}
For each buyer $t$ with value $v_t$ greater than $\phi_{\kappa_T }(\hat y_T - \lfloor \hat y_T \rfloor)$, the fractional allocation satisfies $\hat x_t = 1$.
This is because the total utilization level never exceeds $\hat y_T$, and the marginal price of the resource remains below $v_t$, thus based on Eq.~\eqref{eq:frac:k:risk}, $\hat x_t = 1$. Hence, \algfrac allocates a full unit of the resource to every such buyer.
Moreover, for each buyer with $\hat x_t = 1$, Lemma~\ref{lemma:ppm:frac:relation} implies that, for every realized value of the random seed, \algname also allocates a unit of the item to that buyer.
Thus, across all sample paths of the randomized algorithm, the social welfare of \algname is always incremented by $v_t$.

\textbf{Fact 2.}
For any value $r \in [0,1]$, define the subset of buyers $\hat B^{(r)}_T$ as
\begin{align*} 
\hat B^{(r)}_T = \left\{ t \in [T] \,\middle|\, \hat x_t \not = 0, \ \sum_{\substack{t' < t }} \hat x_{t'} \leq i + r < \sum_{\substack{t' \leq t}} \hat x_{t'} \text{ for some } i \in \{0,1,\dots,k\} \right\}. 
\end{align*}

The set $\hat B^{(r)}_T$ consists of those buyers whose fractional allocation is nonzero and for whom $r$ falls inside the fractional portion of a unit allocated to them. By Lemma~\ref{lemma:ppm:frac:relation}, for every realization $R = r$, a unit of item is allocated to all buyers in the set $\hat B^{(r)}_T$ by \algname.
Furthermore, since the total fractional allocation is $\hat y_T$ and each $\hat x_t \in [0,1]$, the size of $\hat B^{(r)}_T$ satisfies
\begin{equation}
\left|\hat{B}^{(r)}_T\right|
=
\begin{cases}
\lfloor \hat{y}_T \rfloor + 1,
& \text{for } r \in \left[0,\; \hat{y}_T - \lfloor \hat{y}_T \rfloor\right], \\[4pt]
\lfloor \hat{y}_T \rfloor,
& \text{for } r \in \left[\hat{y}_T - \lfloor \hat{y}_T \rfloor,\; 1\right].
\end{cases}
\end{equation}

Thus, when the random seed lies in $[0, \hat{y}_T - \lfloor \hat{y}_T \rfloor]$, the first $\lfloor \hat{y}_T \rfloor + 1$ units are sold at the price levels determined by the pricing functions corresponding to levels $1$ through $\lfloor \hat{y}_T \rfloor + 1$, evaluated at the realized value of $r$.
Similarly, when $r \in [\hat{y}_T - \lfloor \hat{y}_T \rfloor, 1]$, the first $\lfloor \hat{y}_T \rfloor$ units are sold at the price levels determined by the pricing functions corresponding to levels $1$ through $\lfloor \hat{y}_T \rfloor$, again evaluated at the realized value of~$r$.

Putting together Facts 1 and 2, we can lower-bound the \cvar performance of \algname over different realization of random seed $R$ as follows:

\begin{align*}
\alg^{(r)}(I) \geq  
\begin{cases}
\displaystyle 
\int_{0}^{\hat y_T - \lfloor \hat y_T \rfloor}
  \left(\sum_{i=1}^{\lfloor \hat y_T \rfloor + 1} \phi_i(r)\right)dr
\;+\;
\int_{\hat y_T - \lfloor \hat y_T \rfloor}^{1}
  \left(\sum_{i=1}^{\lfloor \hat y_T \rfloor} \phi_i(r)\right)dr
\;+\; \\\\
\hspace{+4.6cm} \displaystyle \sum_{t \in [T]} v_t 
-\phi_{\kappa_T+1}(\hat y_T - \lfloor \hat y_T \rfloor),
& r \in [0,\, \hat y_T - \lfloor \hat y_T \rfloor], \\[1.2em]
\displaystyle 
\int_{\hat y_T - \lfloor \hat y_T \rfloor}^{1}
  \left(\sum_{i=1}^{\lfloor \hat y_T \rfloor} \phi_i(r)\right)dr
\;+\;
\sum_{t \in [T]} v_t 
-\phi_{\kappa_T+1}(\hat y_T - \lfloor \hat y_T \rfloor),
& r \in [\hat y_T - \lfloor \hat y_T \rfloor,\,1].
\end{cases}
\end{align*}

Thus, even in the worst case, we can lower-bound \cvar as follows:
\begin{align*}
\cvar[\alg] & \geq \frac{1}{\delta} \cdot \Bigg( 
\int_{r = 0}^{\max\{0,\delta - (1 - \hat y_T + \lfloor \hat{y}_T \rfloor )\}} 
\sum_{i=1}^{\lfloor \hat{y}_T \rfloor + 1} \phi_i(r) \, d r 
+ \\
 & \qquad \int_{r= \hat y_T - \lfloor \hat{y}_T \rfloor}^{\min\{\delta + \hat y_T - \lfloor \hat{y}_T \rfloor, 1\} } 
\sum_{i=1}^{\lfloor \hat{y}_T \rfloor} \phi_i(r) \, d r + \delta \cdot \sum_{t \in [T]} v_t - \phi_{\kappa_T+1}(\hat y_T - \lfloor \hat y_T \rfloor)
\Bigg).
\end{align*}

It can be verified that, based on the design of the $\phi$ functions given in Theorem~\ref{apx:thm:k:cvar:design:phi}, the right-hand side of the above inequality is exactly
$
\frac{k}{\alphaFDP} \cdot \phi_{\kappa_T+1}(\hat y_T - \lfloor \hat{y}_T \rfloor) + \sum_{t \in [T]} v_t - \phi_{\kappa_T+1}(\hat y_T - \lfloor \hat y_T \rfloor).
$
Therefore,
$
\cvar[\alg] \geq \frac{k \cdot \phi_{\kappa_T+1}(\hat y_T - \lfloor \hat{y}_T \rfloor) + \sum_{t \in [T]} v_t - \phi_{\kappa_T+1}(\hat y_T - \lfloor \hat y_T \rfloor)}{\alphaFDP} \ge \frac{1}{\alphaFDP} \cdot \opt(I), 
$ and the $\alphaFDP$ of \algname over all instances of \oksrisk problem is established.
\end{proof}

\subsection{Proof of Optimality of Design in Theorem~\ref{apx:thm:k:cvar:design:phi}}

\label{apx:prop:order:optimal:k:risk}

In the following, we will prove that for values of $\alpha \ge 1 + \ln(\frac{U}{L})$, according to the pricing design in Theorem~\ref{thm:k:cvar:design:phi}, we will have $\phi_{k}(1) \ge U$, and thus there exists a design that obtains $1 + \ln(\frac{U}{L})$-\cvar competitive. Following, the well-established lower-bound $1+\ln(\frac{U}{L})$, the proof of above theorem follows. 

Let us set $m = \lfloor \frac{k}{\alpha} \rfloor$ and for $i\ge m+2$ define
\begin{align*}
S_i:=\sum_{r=1}^{i}\int_{0}^{\delta}\phi_r(\eta) d\eta,\qquad
b_i =\phi_i(0),\qquad c_i =\phi_i(1).
\end{align*}
Evaluating $\phi_i$ at $x=0$ and $x=1$ gives
\begin{equation}\label{eq:bc}
b_i=\frac{\alpha}{k \delta} S_{i-1},\qquad
c_i=\frac{\alpha}{k \delta} S_{i}\qquad (i\ge m+2),
\end{equation}
and since $\phi_r= L$ for $1\le r\le m$, we have the base mass
\begin{equation}\label{eq:base}
S_{m+1}\ \ge\ \sum_{r=1}^{m}\int_{0}^{\delta}L d\eta\ =\ L\cdot m \delta .
\end{equation}

We claim each $\phi_i$ is nondecreasing on $[0,1]$. This is clear for $i\le m$; for $i\ge m+1$ it follows by induction: if all $\phi_r$ with $r<i$ are nondecreasing, then in case $x+\delta \leq 1$, then we will have:
\begin{align*}
    \phi'_{i}(x) = \frac{\alpha}{k\cdot\delta} \cdot \left( \sum_{r=1}^{i-1} \phi_{r}(x+\delta) - \phi_{r}(x)\right) >0,
\end{align*}
where in above the inequality follows from the induction hypothesis and thus $\phi_{i}(x)$ for $x \in [0,1-\delta]$ is nondecreasing. Furthermore, for $x+\delta > 1$, we will have:
\begin{align*}
    \phi'_{i}(x) = \frac{\alpha}{k\cdot\delta} \cdot \left( \sum_{r=1}^{i} \phi_{r}(x+\delta-1) - \sum_{r=1}^{i-1} \phi_{r}(x)\right) >0,
\end{align*}
where the inequality follows from the fact that $\phi_r(1) = \phi_{r+1}(0)$ and the induction hypothesis. 
Thus, $\phi_{i}(x)$ for $x \in [1-\delta,1]$ is nondecreasing. 

If $g:[0,1]\to\mathbb{R}_{\ge0}$ is nondecreasing, then for every $x\in[0,\delta]$,
\begin{equation}\label{eq:cover}
\int_{x}^{\min\{1, x+\delta\}} g(\eta) d\eta\ +\ \int_{0}^{\max\{0, \delta-1+x\}} g(\eta) d\eta
\ \ \ge\ \ \int_{0}^{\delta} g(\eta) d\eta .
\end{equation}
Indeed, if $x\le 1-\delta$ the second integral vanishes and translating a length-$\delta$ window to the right can only increase the integral of nondecreasing $g$. If $x>1-\delta$, the left side equals $\int_{0}^{1}g-\int_{x+\delta-1}^{x}g$, and among all intervals of length $1-\delta$ the integral of $g$ is minimized on $[0,1-\delta]$, giving $\int_{x+\delta-1}^{x}g\le \int_{\delta}^{1}g$ and hence \eqref{eq:cover}.

Applying Eq. \eqref{eq:cover} to each nondecreasing $\phi_r$ ($r<i$) design given in Theorem~\ref{apx:thm:k:cvar:design:phi}, we obtain for any $x\in[0,\delta]$:
\begin{align*}
\phi_i(x)\ \ge\ \frac{\alpha}{k \delta}\sum_{r=1}^{i-1}\int_{0}^{\delta}\phi_r(\eta) d\eta
\ =\ \frac{\alpha}{k \delta} S_{i-1}\ =\ b_i.
\end{align*}
Integrating over $[0,\delta]$ yields
\begin{equation}\label{eq:Si-growth}
S_i=S_{i-1}+\int_{0}^{\delta}\phi_i(\eta) d\eta\ \ge\ S_{i-1}+\delta b_i
= S_{i-1}+\delta\cdot \frac{\alpha}{k \delta} S_{i-1}
=\Bigl(1+\frac{\alpha}{k}\Bigr)S_{i-1}\qquad (i\ge m+2).
\end{equation}

\paragraph{Putting Everything Together.}
Iterating Eq. \eqref{eq:Si-growth} from $i=m+2$ to $i=k$ and using Eq. \eqref{eq:base} gives
\begin{align*}
S_k\ \ge\ S_{m+1}\Bigl(1+\frac{\alpha}{k}\Bigr)^{k-(m+1)}
\ \ge L \cdot m \delta\ \Bigl(1+\frac{\alpha}{k}\Bigr)^{k-m-1}.
\end{align*}
Thus, by Eq. \eqref{eq:bc}, we have
\begin{align*}
\phi_k(1)=c_k=\frac{\alpha}{k \delta} S_k
\ \ge\ L \cdot \frac{\alpha m}{k} \Bigl(1+\frac{\alpha}{k}\Bigr)^{k-m-1}.
\end{align*}
Since $m=\big\lfloor k/\alpha\big\rfloor$, once $k \rightarrow \infty$, we have $\frac{\alpha m}{k}\to 1$ and
\begin{align*}
\lim_{k \rightarrow \infty} \phi_{k}(1) = \lim_{k \rightarrow \infty} L \cdot \Bigl(1+\frac{\alpha}{k}\Bigr)^{k-m-1}
= \lim_{k \rightarrow \infty} L \cdot \exp \Big((k-m-1)\log(1+\alpha/k)\Big)\ = L \cdot\ e^{ \alpha-1}.
\end{align*}
Thus, for values of $\alpha = 1 + \ln(\frac{U}{L})$, we will have $\phi_{k}(1) \rightarrow U$, satisfying Eq.~\eqref{eq:alpha_fully_dynamic}, and thus there exists a feasible solution for the design of $\{\phi_i\}_{i \in [\Delta+1]}$ for  value of $\alpha  = 1 + \ln(\frac{U}{L})$ satisfying the bound condition in Eq.~\eqref{eq:alpha_fully_dynamic}. On the other hand, since the $1 + \ln(\frac{U}{L})$ is an established lower-bound for the performance of any online algorithm for \oksrisk problem, the optimality of \algname based on the design given in Theorem~\ref{apx:thm:k:cvar:design:phi} follows.